\documentclass[aps,prr,twocolumn,amsfonts, amssymb, amsmath,showpacs,superscriptaddress,showkeys, nofootinbib,longbibliography]{revtex4-1}
\AtBeginDocument{%
\newwrite\bibnotes
\def\bibnotesext{Notes.bib}
\immediate\openout\bibnotes=\jobname\bibnotesext
\immediate\write\bibnotes{@CONTROL{REVTEX41Control}}
\immediate\write\bibnotes{@CONTROL{%
apsrev41Control,author="08",editor="1",pages="1",title="0",year="1"}}
 \if@filesw
 \immediate\write\@auxout{\string\citation{apsrev41Control}}%
\fi
}%
\usepackage[T1]{fontenc}
\usepackage[english]{babel}
\usepackage[utf8]{inputenc}
\usepackage{url}
\usepackage[usenames,dvipsnames]{xcolor}
\usepackage{float}
\usepackage{amsthm}
\newtheorem{proposition}{Proposition}
\newtheorem{lemma}{Lemma}
\usepackage[colorlinks=true, citecolor=blue,urlcolor=blue,linkcolor=blue,filecolor=black]{hyperref}
\usepackage{mathtools}
\usepackage{physics}
\usepackage{graphicx}

\newcommand{\beq}{\begin{equation}}
\newcommand{\eeq}{\end{equation}}
\renewcommand{\emph}{\textit}
\renewcommand{\tr}{\text{Tr}}

\begin{document}

\title{Benchmarking Optical Receivers for Quantum Communication and Randomness Certification}

\author{Hamid Tebyanian}
\email{h.tebyanian@qmul.ac.uk}
\affiliation{School of Physical and Chemical Sciences, Queen Mary University of London, London, E1 4NS, UK}

\begin{abstract}
The choice of optical receiver determines which properties of the transmitted states remain visible in the observed data and therefore affects the performance of different quantum protocols. We compare continuous-variable, photon-counting and hybrid receivers within the same prepare-and-measure framework, using semi-device-independent randomness generation as the main case study. The measurement device is left uncharacterised, while the source is described by the Gram matrix of its pure signal states using an energy-derived overlap constraint, a magnitude-Gram benchmark or the full complex Gram matrix of a certified coherent phase-shift-keyed constellation. Within this framework, the observed receiver statistics are used to bound $H_{\min}(B|X,\Lambda)$ against classical side information correlated with the measurement device but independent of the input. For a fixed Gram matrix, this bound is obtained from an exact semidefinite program, with complex multi-input cases treated in block-real form and checked through the corresponding dual certificate. Photon counting alone is phase blind for fixed-modulus phase encoding and therefore certifies no worst-case randomness. Continuous-variable receivers give the highest certified entropy at moderate energy, while under the nominal source calibration a hybrid receiver performs better at low energy when the beacon-region label is retained in the output. The same receiver statistics also provide receiver-level comparisons for discrete-modulated continuous-variable quantum key distribution, quantum reading, covert communication and quantum-signature verification, without replacing the full security analysis required for each protocol.
\end{abstract}

\maketitle

\section{Introduction}
Optical receivers determine how effectively non-orthogonal states can be distinguished and therefore set the receiver-dependent performance of prepare-and-measure quantum protocols~\cite{Barnett2009,Chefles2000}. The same indistinguishability underlies quantum key distribution~\cite{Gisin2002} and quantum random number generation. In the latter case, bounding what an adversary can guess reduces to bounding how well the preparations can be discriminated under the assumed device constraints. Classical generators cannot provide information-theoretic unpredictability, since a deterministic process is predictable to an adversary who knows its internal state. Quantum generators~\cite{Ma2016,HerreroCollantes2017,Mannalath2023} span a spectrum of trust: device-dependent implementations assume the full inner working of source and detector; device-independent (DI) protocols certify randomness from loophole-free Bell violations alone~\cite{Bell1964,Clauser1969,Pironio2010,Acin2016}, but the required violations~\cite{Hensen2015,Shalm2015} keep their rates far from practical. Semi-device-independent (semi-DI) schemes sit between the two~\cite{Pawlowski2011}. Existing approaches constrain the measurement~\cite{Cao2016}, the source~\cite{Wiseman2009}, a Hilbert-space dimension~\cite{Lunghi2015} or the emitted energy~\cite{VanHimbeeck2017,Rusca2019,Drahi2020}. Prepare-and-measure constructions have also obtained randomness expansion while leaving both source and measurement uncharacterised by trusting a separate testing device and recycling the input randomness~\cite{Bhavsar2026}, while integrated contextuality tests provide an experimentally demonstrated route to semi-DI randomness certification without entanglement~\cite{Genzini2026}. We use the energy assumption because a mean-photon-number cap is directly measurable with a calibrated power monitor and implies a non-trivial floor on pairwise state overlaps, restricting the behaviours available to an adversarial device.

We compare several optical measurements through the conditional laws $P(b|x)$ they induce, including the effects of binning and imperfections. Semi-DI randomness generation provides the working certification problem, with $H_{\min}(B|X,\Lambda)$ obtained through a common semidefinite-programming (SDP) treatment. The resulting receiver hierarchy also supplies receiver-level benchmarks for discrete-modulated CV-QKD, quantum reading, covert communication and quantum-signature verification, as discussed in Sec.~\ref{sec:applications}. We start with a no-go result for phase-insensitive detection:
\begin{proposition}[Phase-blindness no-go]\label{prop:phase-blind} Let the preparations be fixed-modulus coherent states $\hat\rho_x=|\alpha e^{i\phi_x}\rangle\!\langle\alpha e^{i\phi_x}|$ with $|\alpha|^2=\mu$ for all $x$, and let the measurement be any POVM whose elements are diagonal in the Fock basis, i.e.\ $\hat\Pi_b=\sum_{k=0}^{\infty}\pi_{b,k}\,|k\rangle\!\langle k|$
with $\pi_{b,k}\in[0,1]$ and $\sum_b\pi_{b,k}=1$ for all~$k$
(so that $\{\hat\Pi_b\}_b$ is a valid POVM: $\hat\Pi_b\succeq 0$
and $\sum_b\hat\Pi_b=\hat{\mathbb I}$). Then $P(b|x)=P(b)$ for all $b,x$, and consequently $H_{\min}(B|X,\Lambda)=0$ for every $\mu$ and every adversarial model (classical or quantum).
\end{proposition}
\begin{proof} The diagonal statistics are $P(b|x)=\sum_k \pi_{b,k}\,p_k^{(x)}$ with $p_k^{(x)}=|\langle k|\alpha e^{i\phi_x}\rangle|^2 =e^{-\mu}\mu^k/k!$, which is independent of $\phi_x$. Hence $P(b|x)=P(b)$ for all~$x$. Since $\Lambda$ is allowed to be any classical random variable correlated with the device, the adversary may choose $\Lambda$ to be a perfectly correlated copy of the outcome~$B$ (the device internally samples $b\sim P(b)$ and stores the result in $\Lambda$ before outputting~$b$). This deterministic strategy reproduces the data with $P_{\rm guess}=1$ and $H_{\min}=0$. Since a classical adversary is a special case of a quantum one, the no-go holds a fortiori against quantum side information. \end{proof}

Proposition~\ref{prop:phase-blind} concerns certification rather than the physical origin of the outcomes. Since the statistics admit an explicit classical simulation, zero randomness remains the worst case over the adversarial model, irrespective of any randomness that may be present in the physical detector.
Phase-sensitive receivers and feed-forward hybrids instead generate input-dependent statistics and can yield non-zero certified entropy when the observed law $P(b|x)$, together with the overlap floor $\delta\ge (1-2\mu)_+$ or a stronger source model, excludes deterministic convex decompositions in the SDP. The floor, magnitude-Gram and coherent-Gram models also quantify the effect of additional source knowledge. Their certified entropies follow this ordering in every case studied, although only the binary floor and coherent model carry theorem-level worst-case status, since the three models represent distinct assumptions rather than a proved sequence of relaxations.

The guessing-probability optimisation with classical side information is an exact SDP rather than a relaxation whenever the source model fixes the Gram matrix. We solve it for the complex-Hermitian Gram of general $n$-PSK coherent constellations through the block-real embedding. This extends earlier energy-bounded and overlap-based formulations~\cite{Rusca2019,Drahi2020,Tebyanian_2021} to general multi-input constellations. Every reported number is backed by a dual feasibility certificate verified a posteriori, independently of the solver's own report, and the certificate quality is published pointwise rather than assumed; the explicit inner value of the honest realisation is reported alongside, so the gap between honest and adversarial predictability is itself part of the data. Using the same optimisation, we compare five detection families and identify a low-energy regime in which a feed-forward hybrid overtakes homodyne in certified entropy. This advantage is obtained by retaining the beacon-region label in the raw output alphabet, which preserves the phase-sensitive information available to the certification. Building on the energy-bounded demonstrations of Rusca~et~al.~\cite{Rusca2019} and Drahi~et~al.~\cite{Drahi2020}, our analysis extends the framework to a unified comparison of several optical receivers under the same adversarial model.

\section{Framework}
\label{sec:framework}
The schematic representation of the prepare-and-measure scenario is shown in Fig.~\ref{fig:semi-DI_gen}. On input $x\in\{0,\dots,n-1\}$ the preparation stage emits an optical state $\rho_x$ whose mean photon number is bounded, and the measurement stage returns an outcome $b\in\{0,\dots,d-1\}$. Since the measurement device is left uncharacterised, the certification relies only on the observed correlations $p(b|x)$ and the source constraint. $n$ is the input alphabet size, $d$ is the number of outcomes, and $d^n$ is the number of deterministic guess strings considered in the optimisation. The non-orthogonality of the prepared states prevents perfect single-shot discrimination~\cite{Caves1994} and therefore limits the predictability of the outcomes. However, this limitation leads to certified randomness only when the observed data and the source constraint exclude every deterministic convex decomposition. Input dependence of $P(b|x)$ is necessary, since input-independent data can always be reproduced by pre-sampling the outcome, while the SDP determines whether the remaining compatible behaviours are sufficiently constrained to certify non-zero randomness.

\begin{figure}[!h]
\centering
\includegraphics[width=\linewidth]{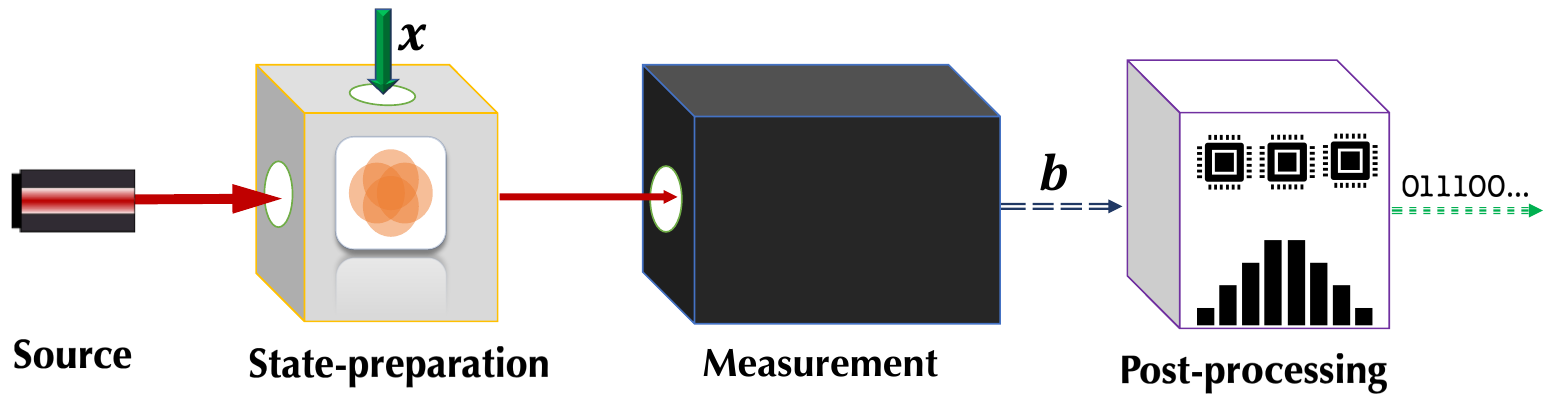}
\caption{General schematic of the protocol. A laser feeds the preparation stage, where the energy of the emitted state is bounded and monitored; the state is measured by an uncharacterised device, and the raw outcomes are fed to post-processing, where the certified min-entropy determines the extractable randomness.}
\label{fig:semi-DI_gen}
\end{figure}

\subsection{Security model and side information}\label{sec:side-info-scope}
The preparation device produces states $\rho_x$ and the measurement device performs a POVM $\hat\Pi_b^\lambda$ selected by a classical variable $\lambda$. We assume that $\Lambda$ is a classical random variable, possibly known to the adversary and correlated with the internal configuration of the measurement device, including thresholds, offsets, splitting ratios, efficiencies, firmware or a classical memory selecting among effective POVMs. It is independent of the fresh input choice, $p(\lambda|x)=p(\lambda)$, and carries no quantum system entangled with the measured mode. Throughout, $B$ denotes the complete retained output in each round. For receivers that keep an auxiliary record for extraction, such as the joint-output hybrid of Sec.~\ref{sec:hybrid-detector}, $B$ is the corresponding tuple, and $H_{\min}(B|X,\Lambda)$ refers to everything supplied to the extractor. Under this assumption the device is, without loss of generality, a convex mixture of POVMs,
\beq
p(b|x) = \sum\limits_\lambda  {p_\lambda}
\tr[\hat\rho_x \hat\Pi_b^\lambda],
\label{eq:constraints}
\eeq
with $\Tr[\hat\rho_x]=1$, $\hat\rho_x\succeq 0$, $\hat\Pi_b^\lambda\succeq 0$ and $\sum_b \hat\Pi_b^\lambda=\mathbb I$, and the certified quantity is $H_{\min}(B|X,\Lambda)$.

This restriction to classical $\Lambda$ is a substantive trust assumption and we state what it excludes. If the adversary holds a quantum register $E$ entangled with the measured mode, the convex-mixture representation fails, the correct description is an instrument acting on the signal, and the relevant quantity becomes $H_{\min}(B|X,E)$, which would require a hierarchy compatible with quantum side information. We do not claim such a result; all numerical certification below is conditioned on the classical-$\Lambda$ model. The impossibility direction holds under any adversarial model. Whenever $P(b|x)$ is independent of $x$, the strategy ``pre-sample $b$ from $P(b)$ and reveal it in $\Lambda$'' reproduces the data with $P_{\rm guess}=1$, giving $H_{\min}=0$ against classical and quantum adversaries alike. This is the mechanism behind Proposition~\ref{prop:phase-blind}.

The energy assumption must itself be operational. If a calibrated monitor returns $\hat\mu$ with calibration uncertainty $\sigma_\mu$ and finite-sample deviation $\Delta_\mu$ at failure probability $\varepsilon_\mu$, the analysis uses, with $\hat n=\hat a^\dagger\hat a$,
\begin{equation}
\mu:=\hat\mu+\sigma_\mu+\Delta_\mu, \;
\Pr\{\forall x:\ \Tr[\hat\rho_x\,\hat n]\le \mu\}\ge 1-\varepsilon_\mu,
\end{equation}
where the per-input tests are combined by a union bound (each at failure probability $\varepsilon_\mu/n$), and $\varepsilon_\mu$ enters the total failure budget alongside $\varepsilon_{\rm stat}$ and $\varepsilon_{\rm PA}$. The working constraint is then
\begin{equation}
\Tr[\hat\rho_x\,\hat n]\le \mu,\qquad \forall x.
\label{eq:energy-bound}
\end{equation}

\subsection{Guessing probability and its exact SDP}
The adversary's figure of merit is the probability of guessing the outcome given the input and $\Lambda$,
\beq
\begin{aligned}
& P_{\rm guess} =  \max \limits_{\{ p_\lambda ,\hat\rho_x,\hat\Pi _b^\lambda\}}
\left(\sum_{x} p_x \sum_{\lambda} p_\lambda
\max_b \bigg\{ \Tr[\hat\rho_x \hat\Pi^\lambda_b] \bigg\}\right), \\
\quad & \text{s.t.\ \eqref{eq:constraints},\;\eqref{eq:energy-bound},\;\text{and POVM/state constraints.}}
\label{eq:P_g}
\end{aligned}
\eeq
We introduce the guess string $\omega=(\omega_0,\dots,\omega_{n-1})\in\{0,\dots,d-1\}^{n}$, where $\omega_x$ is the guess announced for input $x$. Since the inner maximisation factorises over inputs,
\begin{equation}
\sum_x p_x\max_{\omega_x} \Tr[\hat\rho_x\hat\Pi_{\omega_x}^\lambda]
=\max_{\omega\in\{0,\dots,d-1\}^{n}}\sum_x p_x\,\Tr[\hat\rho_x\hat\Pi_{\omega_x}^\lambda],
\end{equation}
every $\lambda$ has an optimal deterministic guess map. Since both the objective and constraints are linear in the POVM elements, the values of $\lambda$ can be grouped by this map, rewriting the optimisation with one POVM per guess string,
\begin{equation}
\overline P_{\rm guess}
=\max_{\{p_\omega,\hat\rho_x,\hat\Pi_b^\omega\}}
\sum_x p_x\sum_\omega p_\omega\,\Tr\!\bigl[\hat\rho_x\,\hat\Pi_{\omega_x}^\omega\bigr],
\label{eq:Pguess-linear}
\end{equation}
subject to the observed-data constraints
\begin{equation}
p(b|x)=\sum_\omega p_\omega\,\Tr[\hat\rho_x\hat\Pi_b^\omega]\qquad \forall b,x.
\label{eq:data-constraints}
\end{equation}
For classical $\Lambda$ this refinement is exact, $\overline P_{\rm guess}=P_{\rm guess}$, whenever the states are common to all branches; the familiar reading of \eqref{eq:Pguess-linear} as an upper-bound relaxation~\cite{Bancal_2014,Brask2017,avesani2021,Tebyanian_2021} is only needed when the states themselves become optimisation variables, as in the energy-only Fock-space model discussed below. The count $d^n$ of guess strings sets the SDP size, not the statistical width: the confidence intervals below concern the $d\cdot n$ frequency bins.

\subsubsection{Fixed-Gram models: an exact program}
\label{sec:gram-sdp}
Suppose the source model fixes the Gram matrix $G_{xx'}=\langle v_x|v_{x'}\rangle$ of pure preparations (the three concrete choices are in Sec.~\ref{sec:source-models}). All ensembles realising a full-rank $G$ are related by a common unitary that can be absorbed into the POVM, so for fixed $G$ it suffices to optimise over POVMs on an $n$-dimensional space. With $\hat M_{b,\omega}:=p_\omega\hat\Pi_b^\omega$ acting on $\mathbb C^{n}$ and $A_x=|v_x\rangle\!\langle v_x|$ built from any factorisation $G=V^\dagger V$,
\begin{equation}
\begin{aligned}
P_{\rm guess}^{G}=\ \max\ & \sum_x p_x\sum_\omega \Tr[A_x \hat M_{\omega_x,\omega}]\\
\text{s.t.}\quad & \hat M_{b,\omega}\succeq 0,\qquad \sum_b \hat M_{b,\omega}=c_\omega\,\mathbb I_n,\\
& c_\omega\ge 0,\qquad \textstyle\sum_\omega c_\omega=1,\\
& \sum_\omega \Tr[A_x \hat M_{b,\omega}]=p(b|x)\quad\forall b,x.
\end{aligned}
\label{eq:gram-sdp}
\end{equation}
Because the states are fixed by $G$, preparation independence holds automatically and \eqref{eq:gram-sdp} gives the exact worst case over all classical-$\Lambda$ strategies compatible with the data and source model, rather than an outer relaxation. Every feasible point is itself a physical realisation, which determines how the comparison between honest and adversarial guessing should be interpreted in Sec.~\ref{sec:results}. For $n\ge3$ coherent constellations $G$ is complex Hermitian, and \eqref{eq:gram-sdp} is solved through the block-real embedding
\begin{equation}
  \Phi(H)=\begin{pmatrix}A&-B\\B&\phantom{-}A\end{pmatrix},
  \; H=A+iB,\;\; H\succeq 0\Leftrightarrow\Phi(H)\succeq 0,
  \label{eq:block-real}
\end{equation}
with $\Re\Tr[A_xM]=\tfrac12\langle\Phi(A_x),\Phi(M)\rangle$. Since every coefficient matrix in \eqref{eq:gram-sdp} is $\Phi$-structured, the symplectic average $S\mapsto\tfrac12(S+JSJ^{\mathsf T})$ with $J=\bigl(\begin{smallmatrix}0&\mathbb I\\-\mathbb I&0\end{smallmatrix}\bigr)$ maps any feasible point of the unstructured real program to a $\Phi$-structured point with the same objective, so the embedding loses nothing. This closes the gap left in our earlier treatment, where the complex case was stated but not solved. All $n\ge3$ coherent-certified values reported here are obtained from this program.

\subsubsection{Dual certificates}\label{sec:dual-cert}
A security claim should not rest on the primal output of a floating-point solver, so every reported bound is accompanied by a verified dual certificate based on the following observation.

\begin{lemma}\label{lem:dual}
Fix any $y\in\mathbb R^{d\times n}$ and Hermitian $W_\omega$ with
$W_\omega\succeq F_b^\omega(y):=\sum_x\bigl(p_x[\omega_x=b]-y_{b,x}\bigr)A_x$ for all $\omega,b$. Then every feasible point of \eqref{eq:gram-sdp} with data in the intervals $|p(b|x)-P(b|x)|\le\Delta$ obeys
\beq
P_{\rm guess}^{G}\ \le\ \sum_{b,x}y_{b,x}P(b|x)+\Delta\sum_{b,x}|y_{b,x}|+\max_\omega \Tr W_\omega .
\eeq
\end{lemma}
\begin{proof}
Write the objective as $\sum_\omega\sum_b\Tr[\hat M_{b,\omega}F_b^\omega]+\sum_{b,x}y_{b,x}\sum_\omega\Tr[A_x\hat M_{b,\omega}]$. The second term is $\sum y_{b,x}\,p(b|x)\le\sum y_{b,x}P(b|x)+\Delta\sum|y_{b,x}|$. In the first, $\Tr[\hat M_{b,\omega}F_b^\omega]\le\Tr[\hat M_{b,\omega}W_\omega]$ by positivity, and summing over $b$ gives $c_\omega\Tr W_\omega\le c_\omega\max_{\omega'}\Tr W_{\omega'}$; summing over $\omega$ with $\sum c_\omega=1$ completes the proof.
\end{proof}

Since Lemma~\ref{lem:dual} holds for any $y$, the multiplier returned by the solver does not need to be trusted directly. We take the numerical dual solution, recompute the minimal dominating operators $W_\omega$, verify $W_\omega-F_b^\omega\succeq-\epsilon\,\mathbb I$ through eigenvalue checks, and absorb any residual $\epsilon$ using $W_\omega\to W_\omega+\epsilon\,\mathbb I$, which raises the bound by $n\epsilon$. The certified value is the right-hand side of the lemma after this repair. The verification is a posteriori and independent of the solver's own report, but it runs in double precision; promoting it to a formal certificate would require interval arithmetic, which we have not implemented. The agreement between the verified bound and the primal value, and the size of the absorbed residuals, are reported with the results (Fig.~\ref{fig:certificates} and Table~\ref{tab:certificates}).\footnote{All programs are solved with MOSEK~11.2.2 through its Fusion interface at interior-point tolerances $10^{-10}$.} 

\subsubsection{Energy-only model}
\label{sec:fock-model}
If one refuses all source structure beyond \eqref{eq:energy-bound}, the states become optimisation variables in an infinite-dimensional space and two prices are paid. First, the shared-state constraint
\begin{equation}
  \frac{\sigma_{x,\omega}}{p_\omega}
  =\frac{\sigma_{x,\omega'}}{p_{\omega'}},
  \qquad \sigma_{x,\omega}:=p_\omega\rho_x,
  \label{eq:prep-indep}
\end{equation}
is bilinear in $(p_\omega,\rho_x)$ and cannot enter a linear program, so it is dropped: each guess string is allowed its own states, the feasible set grows, and the resulting value is a genuine outer bound. This--and not the fixed-Gram program--is where preparation independence is lost. Second, numerics require a truncation $\mathcal H_K=\mathrm{span}\{\lvert0\rangle,\dots,\lvert K\rangle\}$; one then works with moment matrices $\widetilde\Gamma^{(x,\omega)}_{ij}=p_\omega\Tr[\rho_x w_i^{\omega\dagger}w_j^{\omega}]$ over words $w$ of POVM elements of the same context at relaxation level $\ell$, in the spirit of the prepare-and-measure hierarchies of Refs.~\cite{Bancal_2014,Brask2017,avesani2021,Tebyanian_2021,Navascues2008}, with the energy constraint imposed linearly through the moment of the fixed operator $\hat n_K=\sum_{k\le K}k\,|k\rangle\!\langle k|$,
\beq
\sum_{\omega}\widetilde\Gamma^{(x,\omega)}_{\mathbb I,\hat n_K}\le \mu\qquad \forall x .
\eeq
Under the bare energy cap the truncation error must cover arbitrary states, and Markov's inequality is tight in the worst case:
\begin{equation}
\Tr[(\mathbb I-\Pi_K)\hat\rho_x]\le \frac{\mu}{K+1}\equiv \varepsilon_K,\qquad \forall x,
\label{eq:tail-bound}
\end{equation}
which, through the gentle-measurement lemma~\cite{Winter1999}, shifts every Born probability by at most
\begin{equation}
\epsilon_{\rm tr}=\sqrt{\varepsilon_K}+\frac{\varepsilon_K}{2},
\label{eq:eps-tr}
\end{equation}
so the data intervals are widened by $\epsilon_{\rm tr}$ and the certified chain reads
\begin{equation}
  P_{\rm guess}
  \;\le\;\overline P_{\rm guess}
  \;\le\;\overline P_{\rm guess}^{(K,\ell)}+\epsilon_{\rm tr}.
  \label{eq:Pguess-chain}
\end{equation}
This route remains weak at practical truncations. At $K=15$ and $\mu=0.2$, $\varepsilon_K\approx1.25\times10^{-2}$ gives $\epsilon_{\rm tr}\approx0.118$, roughly $400$ times larger than the statistical width at $N_x=10^8$, while $\overline P_{\rm guess}^{(K,\ell)}+\epsilon_{\rm tr}$ exceeds unity and therefore certifies nothing. Reaching $\epsilon_{\rm tr}\sim10^{-2}$ would require $K\sim\mathcal O(10^3)$ according to \eqref{eq:eps-tr}, which reflects the weakness of the model rather than the hierarchy level. For a certified coherent source, the Poisson tail is much smaller. At $\mu=0.5$, for example, $K=10$ gives $\Pr\{N>K\}\approx7.7\times10^{-12}$, whereas Markov would require $K\approx6.5\times10^{10}$. More importantly, the fixed-Gram program \eqref{eq:gram-sdp} acts directly on $\mathbb C^n$ and requires no truncation. All certified results below therefore come from \eqref{eq:gram-sdp} under the source models of Sec.~\ref{sec:source-models}, and we keep \eqref{eq:Pguess-chain} only to make precise what the model-free energy route would cost.

\subsubsection{Finite statistics}
In an experiment $p(b|x)$ is estimated by frequencies $\hat p(b|x)=N_{b,x}/N_x$. With the Hoeffding--union half-widths~\cite{Hoeffding1963}
\begin{equation}
\Delta_{b,x}=\sqrt{\frac{\ln\!\bigl(2\,d\,n/\varepsilon_{\rm stat}\bigr)}{2N_x}},
\label{eq:hoeffding}
\end{equation}
the equalities \eqref{eq:data-constraints} are replaced by intervals $p(b|x)\in[\hat p\pm\Delta_{b,x}]$, which enlarges the feasible set; the POVM constraints project the resulting hyperrectangle back onto the simplex, so no separate care is needed there. The statistical accounting is performed by conditioning on the confidence event. The SDP bound holds on this event, which fails with probability at most $\varepsilon_{\rm stat}$. This failure probability enters the total composable budget of Sec.~\ref{sec:composable} separately and is not identified with a smoothing parameter. At $N_x=10^8$ and $\varepsilon_{\rm stat}=10^{-6}$, $\Delta_{b,x}\approx2.9\times10^{-4}$ for $(d,n)=(3,3)$; propagated through \eqref{eq:gram-sdp} at $\mu=0.2$ this costs the homodyne certificate $1.6\times10^{-3}$~bits, so the asymptotic figures below are representative of realistic block sizes.

\subsection{Composable extraction}\label{sec:composable}
From $N$ rounds with raw string $B^N$ and public inputs $X^N$, two-universal hashing extracts
\begin{equation}
\ell \ \le\ H_{\min}^{\varepsilon_s}(B^N|X^N,\Lambda)\ -\ 2\log_2\!\frac{1}{\varepsilon_{\rm PA}}
\end{equation}
bits that are $\varepsilon_{\rm tot}$-secure with $\varepsilon_{\rm tot}\le \varepsilon_s+\varepsilon_{\rm PA}+\varepsilon_{\rm stat}+\varepsilon_\mu$~\cite{Tomamichel2011LHL,Portmann2022}. Under collective i.i.d.\ attacks, the asymptotic equipartition property~\cite{Tomamichel2009} gives
\begin{equation}
  H_{\min}^{\varepsilon_s}(B^N|X^N,\Lambda)
  \;\ge\;N\,H_{\min}(B|X,\Lambda)
  \;-\;\sqrt{N}\;\Delta_{\rm AEP}(\varepsilon_s,d),
  \label{eq:AEP-certified}
\end{equation}
where the leading term conservatively uses $H_{\min}\le H$ to replace the conditional Shannon entropy by the SDP-certified min-entropy. One substitution is \emph{not} allowed: the observable $H(B|X)$ computed from frequencies obeys, by concavity under the mixture $P(b|x)=\sum_\lambda p_\lambda P(b|x,\lambda)$,
\begin{equation}
  H(B|X)\;\ge\;H(B|X,\Lambda),
  \label{eq:concavity-gap}
\end{equation}
with strict inequality whenever $\Lambda$ is informative--in the extreme of Proposition~\ref{prop:phase-blind}, $H(B|X)>0$ while $H(B|X,\Lambda)=0$--so using it in \eqref{eq:AEP-certified} would overestimate the extractable randomness. The gap is large even in benign cases: at the reference point of Table~\ref{tab:inner-outer-gap}, homodyne shows $H(B|X)=1.37$~bits against a certified $0.34$. For devices with memory across rounds, the product structure used in \eqref{eq:AEP-certified} is no longer available, and an entropy-accumulation or related argument would be required. We therefore restrict the multi-round extraction statement to collective i.i.d.\ attacks.

\subsection{Source models}\label{sec:source-models}
The energy cap alone already bounds overlaps. Writing $\lvert\psi_x\rangle=\sqrt{1-\nu_x}\,\lvert0\rangle+\sqrt{\nu_x}\,\lvert\gamma_x\rangle$ with $\langle0|\gamma_x\rangle=0$,
\begin{equation}
\begin{aligned}
&\bigl|\langle\psi_x|\psi_{x'}\rangle\bigr|
=\Bigl|\sqrt{1-\nu_x}\sqrt{1-\nu_{x'}}+\sqrt{\nu_x\nu_{x'}}\,\langle\gamma_x|\gamma_{x'}\rangle\Bigr|
\ \ge\ \\
&\Bigl(\sqrt{1-\nu_x}\sqrt{1-\nu_{x'}}-\sqrt{\nu_x\nu_{x'}}\Bigr)_+,
\label{eq:overlap-decomp}
\end{aligned}
\end{equation}
and since the vacuum deficit is dominated by the mean, $\nu_x=\sum_{n\ge1}p^{(x)}_n\le\sum_{n\ge1}n\,p^{(x)}_n=\mu_x$, the floor
\begin{equation}
\delta\ \ge\
\Bigl(\sqrt{(1-\mu_x)_+(1-\mu_{x'})_+}
-\sqrt{\mu_x\mu_{x'}}\Bigr)_+
\label{eq:delta-mu-general}
\end{equation}
follows, reducing for equal energies $\mu\le\tfrac12$ to
\begin{equation}
\delta\ \ge\ 1-2\mu,
\label{eq:delta-linear}
\end{equation}
which is tight: $\ket{\psi_{0,1}}=\sqrt{1-\mu}\ket{0}\pm\sqrt{\mu}\ket{1}$ saturates it with $\langle\psi_j|\hat n|\psi_j\rangle=\mu$. For $\mu\ge\tfrac12$ the floor is trivial, so the energy-only analysis is confined to the low-energy regime. For coherent states $\langle\alpha|\beta\rangle=e^{-\mu+\alpha^*\beta}$~\cite{MandelWolf1995}, so BPSK has $\delta=e^{-2\mu}$, and for any $n$-PSK pair the hierarchy
\begin{equation}
  \underbrace{1-2\mu}_{\text{energy floor}}
  \;\le\;
  \underbrace{e^{-2\mu}}_{\text{worst coherent pair}}
  \;\le\;
  \underbrace{e^{-\mu(1-\cos\Delta\phi)}}_{|G^{\rm coh}_{xx'}|}
  \;\le\;1
  \label{eq:overlap-hierarchy}
\end{equation}
holds by $e^{-t}\ge1-t$ and $1-\cos\Delta\phi\le2$. The fixed-Gram models apply to pure signal states because both the overlap derivation and the SDP are defined for a pure-state ensemble. The coherent model adds the assumption that the source emits verified single-mode coherent PSK states. Squeezed-coherent BPSK states have also been studied within a related fixed-Gram semi-DI model with classical detector side information~\cite{TebyanianSqueezed2026}. A mean-photon-number bound alone does not establish purity or determine the Gram matrix, so mixed energy-bounded states are treated only through the truncated-Fock method of Sec.~\ref{sec:fock-model}. Although a fidelity-based extension of \eqref{eq:delta-mu-general} exists for mixed states, it is not used here. The three models are summarised in Table~\ref{tab:models} and Fig.~\ref{fig:hmin}. Their entropy values follow the order floor, magnitude Gram and coherent Gram in all numerical cases studied, but this ordering is not a general inclusion relation.

\paragraph{(i) Energy floor.} Only $|\langle\psi_x|\psi_{x'}\rangle|\ge\delta=(1-2\mu)_+$ is assumed. For $n=2$, the overlap can be chosen real and non-negative by rephasing. We scan the fixed-Gram program over $[\delta,1]$ and find that $\overline P_{\rm guess}$ decreases throughout the feasible range, so the worst case occurs at the overlap floor. At the reference point $\mu=0.2$, overlaps above approximately $0.79$ are already incompatible with the observed data, as shown in Fig.~\ref{fig:scan}a. This agrees with the monotonicity argument of Ref.~\cite{Brask2017} and makes the binary floor result certified. For $n\ge3$, the unknown off-diagonal phases make the modulus constraint non-convex, so an exact convex formulation is not available. We therefore use the real equal-overlap Gram $G_{xx'}=\delta$ and test it against circulant phase families and $30$ random Hermitian phase patterns with the same modulus. In every case, the real Gram gives the largest $\overline P_{\rm guess}$, as shown in Fig.~\ref{fig:scan}b. The corresponding $n\ge3$ floor values are therefore reported as worst-case benchmarks supported by these scans rather than as theorem-level certificates.

\paragraph{(ii) Magnitude Gram.} For fixed-modulus PSK states, discarding the phases of the coherent Gram gives $G^{\rm mag}_{xx'}=e^{-\mu(1-\cos(\phi_{x'}-\phi_x))}$. This matrix is positive semidefinite for every $\mu$ and every phase set: $\mu[\cos(\phi_x-\phi_{x'})]_{xx'}$ is a Gram matrix of rank two, and the entrywise exponential of a PSD matrix is PSD by the Schur product theorem applied to its Hadamard powers, so $G^{\rm mag}=e^{-\mu}\exp_\circ(\mu[\cos\Delta\phi])\succeq0$. This model is neither a restriction nor a relaxation of the coherent Gram model, so $H_{\min}^{\rm mag}$ should be read as an independent benchmark rather than a bound on $H_{\min}^{\rm coh}$. In all numerical results, it lies between the floor and coherent values, as shown in Fig.~\ref{fig:hmin}b.

\paragraph{(iii) Coherent-certified.} A source verified to emit single-mode coherent states fixes the full complex Gram
\begin{equation}
  G^{\rm coh}_{xx'}=\exp\!\bigl[-\mu(1-e^{i(\phi_{x'}-\phi_x)})\bigr],
  \label{eq:coh-gram}
\end{equation}
and using this $G$ in \eqref{eq:gram-sdp}, together with the block-real form \eqref{eq:block-real}, gives the exact worst case under the classical-$\Lambda$ model. This is the source model used for the main certified results.

For the overlap-constrained variant an explicit family with uniform pairwise overlap $\delta$ is occasionally convenient,
\beq
\begin{aligned}
&\ket{\psi_0}=\ket{0}, \qquad
\ket{\psi_1}=\delta \ket{0}+ \sqrt{1-\delta^2} \ket{1}, \\
&\ket{\psi_2}=\delta \ket{0}+ \delta \sqrt{\tfrac{1-\delta}{1+\delta}} \ket{1} + \sqrt{\tfrac{1+\delta-2\delta^2}{1+\delta}} \ket{2},\\
&\ket{\psi_n}= \sum_{i=0}^{n-2} \bra{i}\ket{\psi_{n-1}} \ket{i} + \mathcal{X}_n \ket{n-1}+\mathcal{Y}_n \ket{n}\;\; \forall n > 2,
\label{states}
\end{aligned}
\eeq
with
\beq
\begin{aligned}
&\mathcal{X}_n=\frac{\delta - \sum_{i=0}^{n-2}|\!\braket{i}{\psi_{n-1}}\!|^2}{\braket{n-1}{\psi_{n-1}}},\\
&\mathcal{Y}_n=\sqrt{1-\sum_{i=0}^{n-2}|\!\braket{i}{\psi_{n-1}}\!|^2-\mathcal{X}_n^2},
\end{aligned}
\eeq
constructed inductively so that $\langle\psi_j|\psi_n\rangle=\delta$ for all $j<n$; it realises the equal-overlap Gram but is not used to enforce \eqref{eq:energy-bound}. Note also that no useful bound exists without data: for any $\mu$ or $\delta$ alone, deterministic behaviours are admissible and $H_{\min}=0$, so the role of the source constraint is only to restrict the behaviours compatible with the observed $p(b|x)$, and $H_{\min}(B|X,\Lambda)\le\log_2 d$ is the trivial cap.

\subsection{Receivers}
In this section we describe the five receiver families shown in Fig.~\ref{fig:detection}. For coherent $n$-PSK inputs $\ket{\alpha e^{i\phi_x}}$, with $\phi_x=2\pi x/n$ and $\mu=|\alpha|^2$, each receiver produces a conditional distribution $P(b|x)$ that is used directly in the SDP. The binning and detector imperfections are therefore included in the certification through the corresponding receiver model.

\begin{figure}[!h]
\centering
\includegraphics[width=\linewidth]{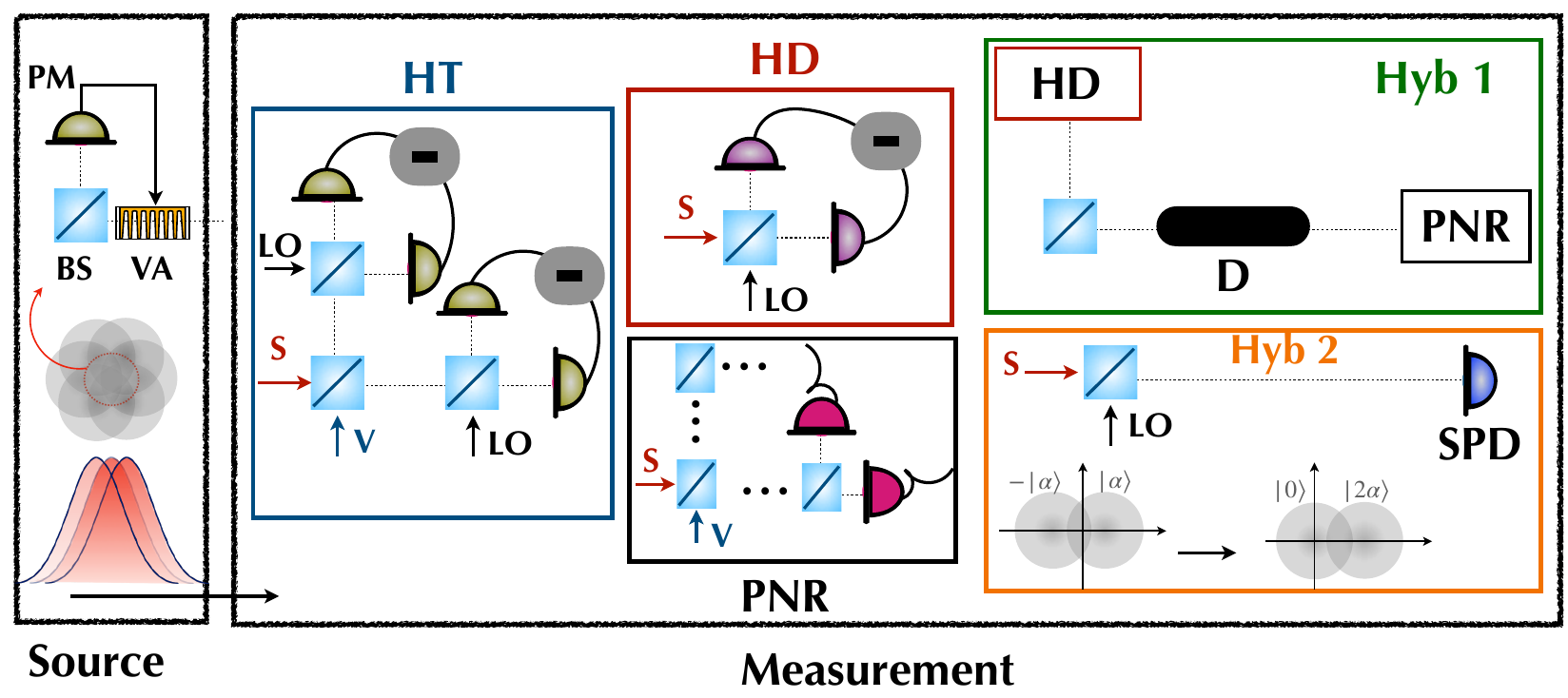}
\caption{Detection schemes, labelled as in the panels. HD: homodyne, one quadrature by interference with a local oscillator. HT: heterodyne, both quadratures at the cost of an added vacuum unit. PNR: photon counting with multiplexed on--off elements. Hyb~1: beacon homodyne, conditional displacement, photon counting. Hyb~2 (Kennedy): nulling displacement and a single-photon detector.}
\label{fig:detection}
\end{figure}

\subsubsection{Homodyne and heterodyne}
With vacuum variance $1/2$, ideal homodyne on the $X$ quadrature returns
\beq
P_X(x\!\mid\!\alpha)=\frac1{\sqrt{\pi}}e^{-\bigl(x-\sqrt{2}\,\Re(\alpha)\bigr)^{2}},
\eeq
and efficiency $\eta_{\rm cv}$ with electronic noise $v_{\rm el}$ acts as the Gaussian channel
\begin{equation}
x_{\rm out}=\sqrt{\eta_{\rm cv}}\;x_{\rm in}+z,\; z\sim\mathcal N(0,\sigma^{2}),\;
\sigma^{2}=\frac{1-\eta_{\rm cv}}{2}+v_{\rm el},
\end{equation}
so the outcome density stays Gaussian with mean $\sqrt{2\eta_{\rm cv}}\,\Re(\alpha)$ and variance $\tfrac12+v_{\rm el}$. Binning the axis into $[A,B]$ gives
\beq
\begin{aligned}
\label{prob_homo}
&P_{[A,B]}^{\rm HD}(\alpha)= \\
&\frac12\!\left[
\erf\!\left(\frac{B-\sqrt{2\eta_{\rm cv}}\,\Re(\alpha)}{\sqrt{1+2v_{\rm el}}}\right)
-\erf\!\left(\frac{A-\sqrt{2\eta_{\rm cv}}\,\Re(\alpha)}{\sqrt{1+2v_{\rm el}}}\right)
\right].
\end{aligned}
\eeq
For $n$-PSK we measure the $P$ quadrature, whose means $\sqrt{2\eta_{\rm cv}\mu}\sin\phi_x$ separate all inputs for odd $n$, and place the $d=n$ bin edges at the midpoints of adjacent nominal means, which is the MAP rule for equal-variance Gaussians with uniform priors; the edges are frozen at their nominal positions in all drift studies.

Heterodyne implements the coherent-state POVM $\{\pi^{-1}\ket{\beta}\!\bra{\beta}\}$, an Arthurs--Kelly joint measurement~\cite{ArthursKelly1965}. The same two-quadrature receiver architecture has been used for source-device-independent randomness generation with real-time FPGA extraction~\cite{Cizauskas2026}, and its ideal outcome is distributed as
\beq
p_{\rm HT}(\beta\!\mid\!\alpha)=\frac1\pi
           e^{-|\,\beta-\alpha|^{2}},\quad\beta\in\mathbb C,
\eeq
each quadrature carrying the unavoidable extra $3$~dB relative to homodyne on the same energy--no uncertainty-principle violation is involved. In Cartesian and polar form,
\beq
\label{probhet}
P_{\mathcal R}^{\rm HT}(\alpha)=
\frac1\pi\!
\int_{b_i}^{b_f}\!\!\int_{B_i}^{B_f}
e^{-(b-a)^{2}-(B-A)^{2}}\;db\,dB,
\eeq
\beq
\label{probhet2}
P_{\mathcal S}^{\rm HT}(\alpha)=
\frac{e^{-a^{2}}}{\pi}
\int_{r_i}^{r_f}\!\!\int_{\theta_i}^{\theta_f}
      r\,e^{-r^{2}+2ar\cos(\theta-\phi)}\,dr\,d\theta,
\eeq
and the radial integral over $[0,\infty)$ has the closed form
\begin{equation}
\int_0^\infty r\,e^{-r^2+2cr}\,dr
=\tfrac12+\tfrac{\sqrt\pi}{2}\,c\,e^{c^2}[1+\erf(c)],\;
c=a\cos(\theta-\phi).
\end{equation}
For $n$-PSK we use $d=n$ angular sectors of width $2\pi/n$ centred on the nominal phases, likewise frozen under drift.

\subsubsection{Photon counting}
A single-photon detector with efficiency $\eta$ and dark-count probability $p_{\rm dc}$ has the Fock-diagonal POVM
\begin{equation}
\hat\Pi_{0}=(1-p_{\mathrm{dc}})\sum_{k=0}^{\infty}(1-\eta)^{k}\ket{k}\!\bra{k},\qquad
\hat\Pi_{1}=\hat{\mathbb I}-\hat\Pi_{0},\label{1}
\end{equation}
so for coherent input $P(0\mid x)=(1-p_{\mathrm{dc}})\,e^{-\eta\mu_x}$. Multiplexed photon-number resolution~\cite{Lita2008,Natarajan2012} distributes the mode over $N_{\mathrm{bin}}$ on--off elements; for coherent input the fan-out preserves product coherence, each bin sees an independent coherent state of mean $\mu_x/N_{\rm bin}$, and the click number is binomial,
\begin{equation}
\begin{aligned}
&P(c\mid x)=\binom{N_{\mathrm{bin}}}{c}\Bigl[1-s_x\Bigr]^{c}s_x^{\,N_{\mathrm{bin}}-c},\\
&s_x=(1-p_{\mathrm{dc}}^{\mathrm{bin}})\,e^{-\eta\mu_x/N_{\mathrm{bin}}},
\label{5}
\end{aligned}
\end{equation}
with $p_{\rm dc}^{\rm bin}$ the per-element dark-count probability. Both models are diagonal in the Fock basis, so $\partial P/\partial\phi_x\equiv0$ for fixed-modulus inputs: photon counting alone satisfies the hypothesis of Proposition~\ref{prop:phase-blind} and certifies nothing for PSK, however good the detector. Figure~\ref{fig:ternary_pbx_mu02} shows this at the level of the actual SDP inputs. Under a trusted-overlap model at higher $\mu$, photon counting can still exploit multi-photon statistics for intensity-modulated alphabets, which is why we keep it in the comparison as the on--off keying reference.

\begin{figure}[!h]
\centering
\includegraphics[width=\linewidth]{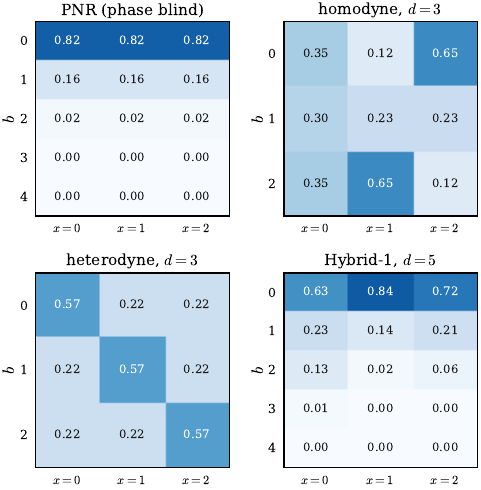}
\caption{Conditional laws $P(b|x)$ for ternary PSK at $\mu=0.2$, ideal detectors--the exact matrices used as SDP input. PNR (multiplexed, $N_{\rm bin}=20$, click bins $\{0\},\{1\},\{2,3\},\{4$--$7\},\{\geq8\}$) is phase blind: its columns are identical, so $H_{\min}=0$ by Proposition~\ref{prop:phase-blind}. Homodyne ($P$ quadrature, midpoint bins), heterodyne (angular sectors) and Hybrid-1 (parameters of Sec.~\ref{sec:hyb1-params}, beacon label traced out) produce input-dependent columns.}
\label{fig:ternary_pbx_mu02}
\end{figure}

\subsubsection{Hybrid-1: beacon homodyne with feed-forward displacement}\label{sec:hybrid-detector}
Hybrid receivers combine quadrature information with photon counting~\cite{Izumi2012,Becerra2013}. Hybrid-1 (Fig.~\ref{fig:hyb}) taps the input on a 50:50 splitter; the reflected mode is measured by homodyne on the $p$ quadrature, whose outcome density for input $x$ is
\begin{equation}
f_x(p)=\frac{1}{\sqrt{\pi}}\,
\exp\!\Bigl[-\bigl(p-|\alpha|\sin\phi_x\bigr)^2\Bigr],
\end{equation}
since the tapped amplitude is $\alpha e^{i\phi_x}/\sqrt2$. Two thresholds $T_2<T_1$ cut the axis into $\Omega_2=(-\infty,T_2)$, $\Omega_0=[T_2,T_1]$, $\Omega_1=(T_1,\infty)$, and on $p\in\Omega_k$ the transmitted mode is displaced by $D(kg\alpha)$ with gain $g>0$, giving the displaced amplitude and mean photon number
\beq
\begin{aligned}
&\beta_{k,x}=\frac{\alpha e^{i\phi_x}}{\sqrt2}+kg\alpha,
\\
&\lambda_{k,x}=|\beta_{k,x}|^2
=\mu\Bigl(\tfrac12+k^{2}g^2+\sqrt2\,kg\cos\phi_x\Bigr).
\label{eq:lambda-def}
\end{aligned}
\eeq
The displaced mode is measured by the multiplexed click detector \eqref{5} with $\mu_x\to\lambda_{k,x}$, and a deterministic map coarse-grains the click number $c$ into bins $\mathcal C_0=\{0\}$, $\mathcal C_1=\{1,\dots,c_1^{\max}\}$, \dots, $\mathcal C_4=\{c_3^{\max}{+}1,\dots,N_{\rm bin}\}$. Two output conventions must be distinguished, because they certify different things. If only the click bin is kept as the outcome, the law is
\begin{equation}
\begin{aligned}
P(b\mid x)=\sum_{k}\;
\int_{\Omega_k} f_x(p)\,dp
\sum_{c\in\mathcal C_b}
\binom{N_{\mathrm{bin}}}{c}(1-s_{k,x})^{c} s_{k,x}^{\,N_{\mathrm{bin}}-c},
\end{aligned}
\label{eq:hyb-coarse}
\end{equation}
$s_{k,x}=(1-p^{\rm bin}_{\rm dc})e^{-\eta\lambda_{k,x}/N_{\rm bin}}$, a $d=5$ alphabet. If the beacon-region label $k$ is instead retained as part of the raw outcome, the alphabet is the joint $(k,b)$ and the integrand is not summed over $k$. The certified quantity for that alphabet is $H_{\min}(K,B|X,\Lambda)$: the min-entropy of the pair, which stays secret and is fed whole into the extractor. This must not be confused with the case of a beacon made public, where the relevant quantity would be $H_{\min}(B|X,\Lambda,K)$ with $K$ handed to the adversary--a smaller quantity that we do not compute. Since deterministic post-processing cannot increase conditional min-entropy, $H_{\min}(B_{\rm click}|X,\Lambda)\le H_{\min}(K,B_{\rm click}|X,\Lambda)$, and the gap turns out to be the whole low-energy advantage of this receiver (Sec.~\ref{sec:results}): tracing out the beacon discards precisely the phase-sensitive record that makes the hybrid competitive.

The beacon breaks a structural degeneracy. For ternary PSK, $\cos(2\pi/3)=\cos(4\pi/3)=-\tfrac12$ makes $\lambda_{k,1}=\lambda_{k,2}$ for every $k$, so conditioned on a region the click statistics of $x=1$ and $x=2$ are identical; the beacon means $\pm|\alpha|\sqrt3/2$ separate them, and the joint outcome distinguishes all three inputs. The 50:50 tap also imposes a real cost: the transmitted mode carries $\mu/2$, so the undisplaced pairwise overlap $e^{-\frac{\mu}{2}(1-\cos\Delta\phi)}$ is larger than the unsplit one, and the hybrid can only win when the displacement-induced separation of the $\lambda_{k,x}$ plus the enlarged outcome alphabet buys back more than this 3~dB penalty--which happens in the energy-starved regime, not at moderate $\mu$.

\begin{figure*}[!t]
\centering
\includegraphics[width=\linewidth]{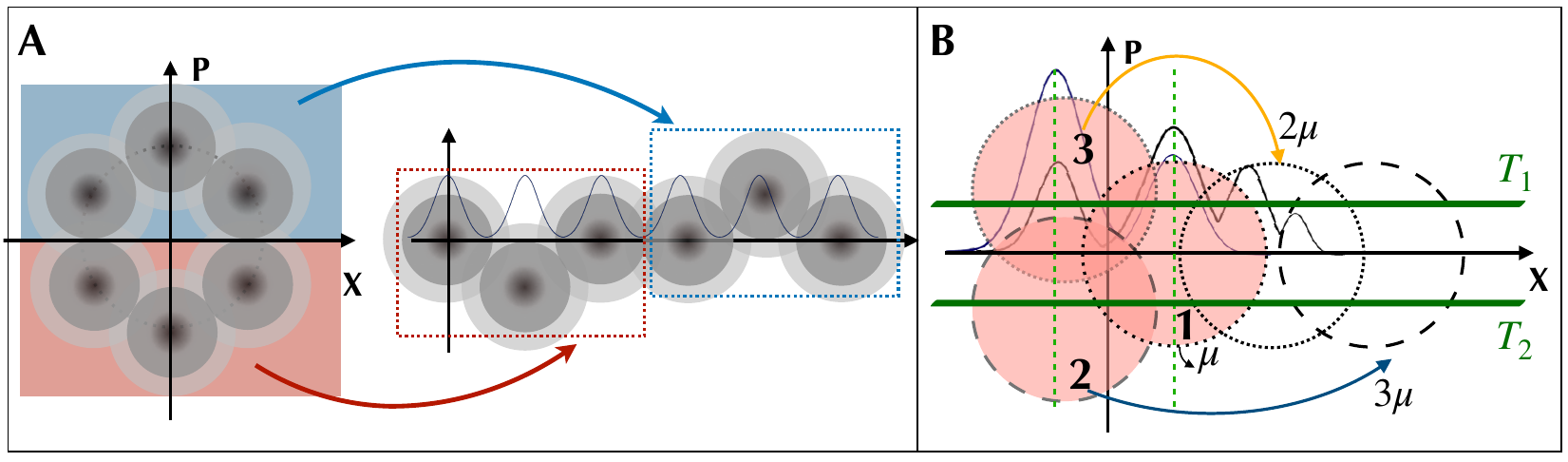}
\caption{Hybrid-1 receiver. (a) General operating principle. A beacon homodyne measurement provides a coarse phase estimate, which determines the conditional displacement applied to the transmitted mode before photon-number-resolving detection. The panel gives a schematic phase partition rather than the precise ternary implementation. (b) Ternary implementation. Two thresholds, $T_2<T_1$, divide the beacon outcome into three regions associated with $D(2g\alpha)$, the identity operation and $D(g\alpha)$. The displaced amplitudes produce the branch- and input-dependent mean photon numbers $\lambda_{k,x}$ of Eq.~\eqref{eq:lambda-def}, which determine the corresponding PNR click distributions.}
\label{fig:hyb}
\end{figure*}

\subsubsection{Hybrid-2: Kennedy receiver and its cascade}
The Kennedy receiver~\cite{Kennedy1973,Dolinar1973} interferes the BPSK signal with a phase-locked local oscillator on a highly transmissive splitter, $\hat a_{\text{out}} = t\,\hat a_{s}+r\,\hat a_{\text{LO}}$, choosing $\beta=t\alpha/r$ so that $-\alpha$ is nulled and $+\alpha$ is displaced to $\gamma_+=2t\alpha$. An SPD with efficiency $\eta$ and dark counts $p_{\rm dc}$ then gives
\begin{equation}
\begin{aligned}
P(b{=}1\mid +)&=1-(1-p_{\text{dc}})\exp\!\bigl[-4\eta|t|^{2}\mu\bigr],\\
P(b{=}1\mid -)&=p_{\text{dc}},
\end{aligned}
\label{eq:kennedy-binary}
\end{equation}
the ideal case being the displaced on--off POVM $\{D^\dagger(\alpha)\ket{0}\!\bra{0}D(\alpha),\ \mathbb I-D^\dagger(\alpha)\ket{0}\!\bra{0}D(\alpha)\}$. The factor $4$ is the coherent displacement $|2t\alpha|^2$, not an interference gain. For $n$-PSK the idea extends to a cascade: stage $k$ taps a fraction $|r|^2$ of the (depleted) signal $\mu_k=\mu(1-|r|^2)^k$, nulls hypothesis $k$, and the output is the first stage that clicks, with $b=n{-}1$ if none does. The stage amplitudes and no-click probabilities are
\begin{equation}
  |\gamma_{kx}|^2
  =4|r|^2\mu_k\sin^2\!\tfrac{\pi(x-k)}{n},\quad
  q_{kx}=(1-p_{\rm dc})e^{-\eta|\gamma_{kx}|^2},
  \label{eq:kennedy-qkx}
\end{equation}
\begin{equation}
\begin{aligned}
P(b{=}k\mid x)&=(1-q_{k x})\prod_{\ell=0}^{k-1}q_{\ell x},\\
P(b{=}n{-}1\mid x)&=\prod_{\ell=0}^{n-2}q_{\ell x}.
\end{aligned}
\end{equation}
The label $b$ therefore records the first \emph{rejection}: a click at stage $k$ is evidence against hypothesis $k$, not a decoded estimate, which is what the passive weak-tap cascade physically provides. It should not be confused with the adaptive accept-on-no-click rule of Ref.~\cite{Becerra2013}, whose stages consume large fractions of the pulse. For everything computed in this paper the announcement convention is immaterial in any case: $I(X;B)$, the MAP value and the SDP optimum are all invariant under relabelling of the outcomes, so any deterministic naming of the click record yields identical numbers. For ternary PSK, $\sin^2(\pi/3)=\sin^2(2\pi/3)=3/4$ makes every non-nulled input produce identical click probabilities at each stage; this architectural degeneracy, not energy loss, is why the cascade is the weakest ternary receiver: at the operating tap $|r|^2=0.10$ the final stage still receives $90\%$ ($n{=}3$) and $73\%$ ($n{=}5$) of the signal, a depletion far too small to account for the observed photon-efficiency deficit, yet each stage can only reject one hypothesis.

\subsubsection{Numerical parameters}\label{sec:hyb1-params}
All figures and SDP runs use the following, stated here once. Ternary Hybrid-1 at the reference point $\mu=0.2$ (Table~\ref{tab:inner-outer-gap}): $T_2=-0.43$, $T_1=+0.43$, $g=1$, $N_{\rm bin}=20$, $\eta=1$, $p_{\rm dc}^{\rm bin}=0$, click cuts $(c_1^{\max},c_2^{\max},c_3^{\max})=(1,3,7)$, beacon label traced out ($d=5$). Low-energy Hybrid-1: thresholds at the midpoints of the nominal beacon means (which scale as $\sqrt\mu$), gain optimised--$g\approx4$ is near-optimal at small $\mu$, and the older impression that the nulling choice $g\approx1$ suffices does not survive the numbers, costing about $15\%$ of the photon efficiency--and, where stated, the joint output $(k,b)$ with click cuts $(1)$, i.e.\ $c\in\{0\},\{1\},\{\ge2\}$, giving $d=9$ for the SDP. Kennedy: binary $|t|^2=0.99$; the ternary and 5-PSK cascades use $|r|^2=0.10$ (the tap that reproduces the reference values; at $|r|^2=0.01$ the cascade barely clicks at these energies and certifies an order of magnitude less). Detector-imperfection studies vary $\eta$ and $p_{\rm dc}$ as stated in captions. For the common-efficiency comparison of Fig.~\ref{fig:eta}, a pure-loss channel of transmissivity $\eta$ is applied before the receiver, so that $\alpha\rightarrow\sqrt{\eta}\alpha$ in the corresponding conditional law; internal detector efficiencies are otherwise kept fixed to avoid double counting.

\section{Distinguishability and its relation to certification}\label{sec:distinguishability}
To compare receivers as decoders we treat each as a classical channel with MAP success and calibrated index
\begin{equation}
P_{\mathrm{suc}}=\sum_b \max_x \,p_x\,P(b|x),
\qquad
\mathcal{I}=\frac{n\,P_{\mathrm{suc}}-1}{n-1},
\label{eq:distinguish-index}
\end{equation}
so $\mathcal I=0$ is blind guessing and $\mathcal I=1$ perfect discrimination. The binary quantum limit is the Helstrom bound~\cite{Helstrom1976}
\begin{equation}
\begin{aligned}
&P_{\mathrm{err}}^{\mathrm{Helstrom}}
=\frac12\Bigl(1-\sqrt{1-|\langle\psi_0|\psi_1\rangle|^2}\Bigr)
\;\xrightarrow{\ \rm BPSK\ }\\
&\frac12\Bigl(1-\sqrt{1-e^{-4\mu}}\Bigr),
\label{eq:Helstrom-BPSK}
\end{aligned}
\end{equation}
whose small-$\mu$ expansion is $\tfrac12-\sqrt{\mu}+O(\mu)$, while at large $\mu$ it approaches $\tfrac14e^{-4\mu}$. The Gaussian receivers obey
\begin{equation}
P_{\mathrm{err}}^{\mathrm{HD}}(\mu)=\tfrac12\,\mathrm{erfc}\!\bigl(\sqrt{2\mu}\bigr),
\qquad
P_{\mathrm{err}}^{\mathrm{HT}}(\mu)=\tfrac12\,\mathrm{erfc}\!\bigl(\sqrt{\mu}\bigr),
\end{equation}
the heterodyne penalty being exactly the joint-measurement 3~dB, while the ideal Kennedy error $\tfrac12 e^{-4\eta|t|^2\mu}$ beats the homodyne exponent at large $\mu$. Photon counting on OOK $\{\ket 0,\ket\alpha\}$ gives the symmetric error $\tfrac12[p_{\mathrm{dc}}+(1-p_{\mathrm{dc}})e^{-\eta\mu}]$, and the multiplexed PNR reduces to the same form with $p_{\mathrm{dc}}^{\mathrm{det}}=1-(1-p_{\mathrm{dc}}^{\mathrm{bin}})^{N_{\mathrm{bin}}}$.

These benchmarks must not be conflated with certification. $\mathcal I(\mu)$ is a property of the known physical channel, whereas $H_{\min}(B|X,\Lambda)$ comes from an adversarial optimisation over every realisation compatible with the same data and source model; the two connect only through the data constraints. Input-independent statistics force $H_{\min}=0$ by Proposition~\ref{prop:phase-blind}, so $\mathcal I>0$ is necessary for certification. The converse does not hold, since coarse binning can leave the SDP enough freedom to reproduce $P(b|x)$ with near-deterministic mixtures even when the physical channel is far from deterministic. Every randomness claim below therefore rests on the SDP, while the distinguishability plots provide receiver benchmarks and identify obvious no-go regimes. For each receiver we report the certified $H_{\min}=-\log_2\overline P_{\rm guess}$ under the stated source model, the peak $H^\star_{\min}$ and its maximiser $\mu^\star$ where relevant, and the honest inner value $\underline P_{\rm guess}=\sum_xp_x\max_bP(b|x)$ of the trivial-$\Lambda$ realisation, whose distance from $\overline P_{\rm guess}$ measures how much the allowed side information is worth to the adversary.

\section{Results}\label{sec:results}

\subsection{Certified min-entropy}
Table~\ref{tab:inner-outer-gap} gives the reference point: ternary PSK, $\mu=0.2$, ideal detectors, coherent-certified model. Homodyne certifies $0.336$~bits, heterodyne certifies $0.328$~bits, while the Kennedy cascade and coarse-grained Hybrid-1 certify $0.103$ and $0.102$~bits, respectively. The inner value is not a lower bound on the certified entropy; it is the guessing probability of the honest device with trivial side information, and the gap $\overline P_{\rm guess}-\underline P_{\rm guess}$ is the adversary's side-information advantage. Because the fixed-Gram SDP is exact, this gap is not relaxation slack--every feasible point is a physical realisation--so its size is a genuine statement about the model: for homodyne, branch-dependent POVM mixtures predict $0.24$ better than naive MAP decoding, which in entropy terms halves the naive estimate ($0.86\to0.34$~bits).

\begin{table}[h]
\centering
\caption{Ternary PSK, $\mu=0.2$, ideal detectors, coherent-certified model (exact SDP). $\overline{P}_{\rm guess}$: verified dual bound; $\underline{P}_{\rm guess}=\sum_x p_x\max_b P(b|x)$: honest MAP value of the trivial-$\Lambda$ realisation. The gap is the adversary's side-information advantage, not relaxation slack: for phase-sensitive receivers, allowed branching over $\Lambda$ is worth $\approx0.2$ in guessing probability (about half a bit), while the near-deterministic Kennedy channel leaves almost nothing to exploit.}
\label{tab:inner-outer-gap}
\begin{tabular}{lcccc}
\hline
Receiver & $\overline{P}_{\rm guess}$ & $\underline{P}_{\rm guess}$
& Gap ($P_{\rm g}$) & Gap ($H_{\min}$, bits) \\
\hline
Homodyne ($d{=}3$)   & 0.7921 & 0.5502 & 0.2419 & 0.526 \\
Heterodyne ($d{=}3$) & 0.7964 & 0.5668 & 0.2296 & 0.491 \\
Kennedy cascade      & 0.9310 & 0.9272 & 0.0037 & 0.006 \\
Hybrid-1 ($d{=}5$)   & 0.9320 & 0.7291 & 0.2029 & 0.354 \\
\hline
\end{tabular}
\end{table}

\begin{table*}[t!]
\centering
\caption{Scope of the claims: source models, their status, and where each is used. The adversary model is classical $\Lambda$ (Sec.~\ref{sec:side-info-scope}) in every row; the SDP is exact for every fixed-Gram row. Certified means a theorem-level worst case under that model; the $n\ge3$ floor is a benchmark whose extremality is supported by the phase scans of Fig.~\ref{fig:scan}.}
\label{tab:models}
\resizebox{\textwidth}{!}{%
\begin{tabular}{lcccc}
\hline
Model & constraint & truncation & certified & used in \\
\hline
Energy floor, $n{=}2$ & $G_{01}\in[(1-2\mu)_+,1]$ & none & yes & Fig.~\ref{fig:scan}a \\
Energy floor, $n{\ge}3$ & $G_{xx'}=(1-2\mu)_+$ & none & benchmark & Figs.~\ref{fig:hmin}b,~\ref{fig:eta},~\ref{fig:scan}b \\
Magnitude Gram & $G=G^{\rm mag}$, Sec.~\ref{sec:source-models}(ii) & none & separate model & Fig.~\ref{fig:hmin}b \\
Coherent-certified & $G=G^{\rm coh}$, Eq.~\eqref{eq:coh-gram} & none & yes & Tables~\ref{tab:inner-outer-gap},~\ref{tab:certificates}; Figs.~\ref{fig:hmin}--\ref{fig:certificates} \\
Energy-only Fock & $\Tr[\rho_x\hat n]\le\mu$ & $K$, Eq.~\eqref{eq:eps-tr} & yes, trivial at $K{=}15$ & Sec.~\ref{sec:fock-model} only \\
\hline
\end{tabular}}
\end{table*}

Figure~\ref{fig:hmin}(a) shows the certified entropy against $\mu$ in the coherent model. Homodyne rises from $0.35$~bits at $\mu=0.005$ to its peak $H^\star_{\min}=0.43$~bits at $\mu^\star\approx0.054$ and falls back to $0.09$~bits at $\mu=1$, where the constellation becomes too distinguishable to constrain the adversary; certification does not die at low energy because the bin edges scale with $\sqrt\mu$, so the states approach each other exactly as fast as the data flatten. Heterodyne runs parallel and slightly below, while the Kennedy cascade and the coarse-grained Hybrid-1 remain substantially below the Gaussian receivers. The joint-output Hybrid-1 changes the low-energy picture: at $\mu=0.01$ it certifies $0.379$~bits against homodyne's $0.362$, the advantage reaching $7\%$ at $\mu=0.005$, while by $\mu\approx0.02$ homodyne is ahead again. Retaining the beacon label in the raw output is what buys this--the coarse-grained $d=5$ output never beats homodyne anywhere. The pipeline also carries the larger alphabet: for 5-PSK, with $5^5=3125$ guess strings in the program, coherent-certified homodyne reaches $H_{\min}=0.50$, $0.57$ and $0.60$~bits at $\mu=0.1$, $0.2$ and $0.4$ under the same dual verification, so a larger constellation certifies substantially more entropy per round at comparable energy.

Because the low-energy crossover is central to the receiver comparison, we test whether it survives finite statistics. Table~\ref{tab:finite} repeats the certification with Hoeffding boxes \eqref{eq:hoeffding} on all $d\cdot n$ frequency bins at $\varepsilon_{\rm stat}=10^{-6}$. The joint alphabet pays about four times the finite-size penalty of the $d=3$ homodyne alphabet, with $0.026$ against $0.006$~bits at $\mu=0.005$ and $N_x=10^8$, in part because it has three times as many frequencies to stabilise, yet the ordering survives: at $N_x=10^8$ the hybrid still certifies more than homodyne at $\mu=0.005$ and $0.01$, with the margin compressed to about $1.5\%$, and at $N_x=10^9$ the margin recovers to $3.5$--$4.5\%$; at $\mu=0.02$ homodyne leads in every column, consistent with the asymptotic crossover. Energy-calibration uncertainty must be evaluated at fixed observed data by varying the allowed Gram matrix. A $1\%$ upward shift in the certified $\mu$ reduces the homodyne bound by about $0.005$--$0.007$~bits over the tested points, so the low-energy hybrid crossover is not established under this uncertainty. The failure probability $\varepsilon_\mu$ enters the total budget as in Sec.~\ref{sec:side-info-scope}.

\begin{table}[h]
\centering
\caption{Finite-statistics check of the low-energy hybrid advantage (coherent model, ideal detectors, $\varepsilon_{\rm stat}=10^{-6}$). Certified $H_{\min}$ in bits; ``asym.'' uses exact data constraints, the $N_x$ columns use Hoeffding boxes on all $d\cdot n$ bins ($\Delta_{b,x}=2.9\times10^{-4}$ and $3.0\times10^{-4}$ at $N_x=10^8$ for $d=3$ and $d=9$; $\sqrt{10}$ smaller at $10^9$). Bold marks the larger certified value at each $(\mu,N_x)$.}
\label{tab:finite}
\resizebox{\columnwidth}{!}{%
\begin{tabular}{lcccccc}
\hline
 & \multicolumn{3}{c}{Homodyne ($d{=}3$)} & \multicolumn{3}{c}{Hybrid-1 joint ($d{=}9$)} \\
$\mu$ & asym. & $N_x{=}10^8$ & $N_x{=}10^9$ & asym. & $N_x{=}10^8$ & $N_x{=}10^9$ \\
\hline
0.005 & 0.347 & 0.341 & 0.345 & \textbf{0.371} & \textbf{0.345} & \textbf{0.361} \\
0.01  & 0.362 & 0.357 & 0.360 & \textbf{0.379} & \textbf{0.362} & \textbf{0.373} \\
0.02  & \textbf{0.385} & \textbf{0.382} & \textbf{0.384} & 0.362 & 0.352 & 0.359 \\
\hline
\end{tabular}}
\end{table} Figure~\ref{fig:hmin}(b) isolates the value of source knowledge for homodyne: the floor, magnitude and coherent models certify $0.119$, $0.266$ and $0.336$~bits at $\mu=0.2$, the floor curve dying at $\mu=\tfrac12$ where $(1-2\mu)_+$ becomes trivial. Sensitivity to the common loss is modest for both Gaussian receivers (Fig.~\ref{fig:eta}). At $\mu=0.2$ they each certify approximately $0.17$~bits at $\eta=0.3$. Heterodyne is slightly above homodyne over most of the lossy region, while homodyne recovers the lead near unit transmissivity; the same interchange appears in the energy-floor benchmark.

\begin{figure}[!h]
\centering
\includegraphics[width=\linewidth]{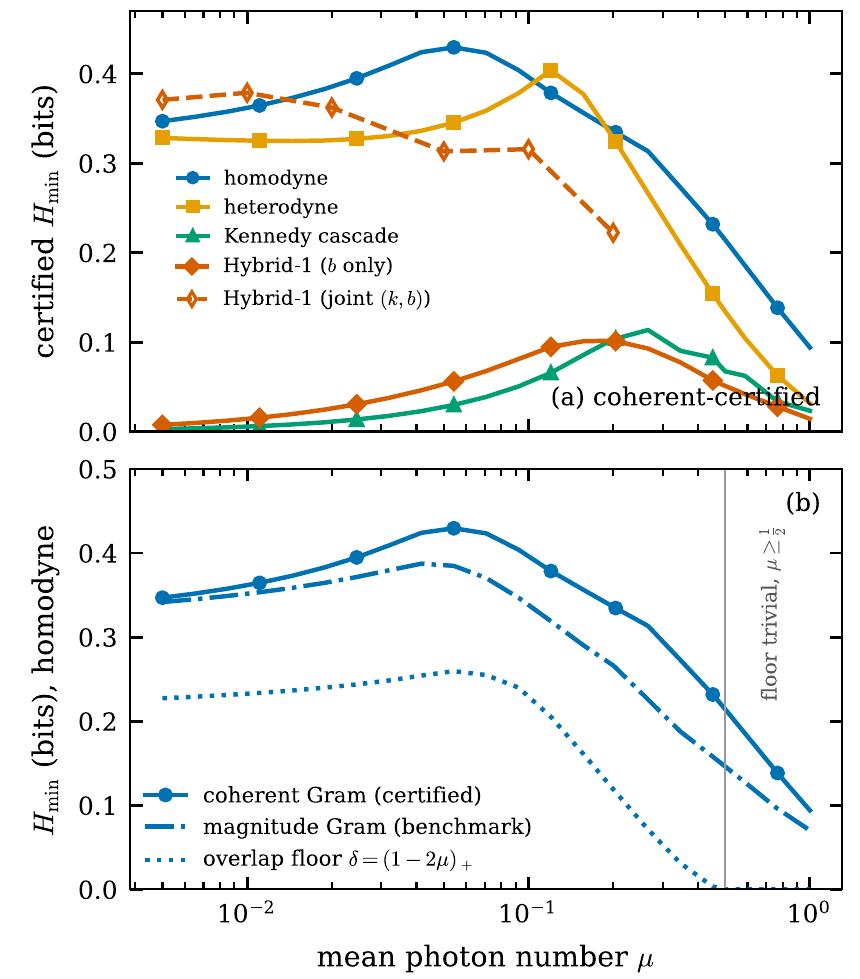}
\caption{Certified min-entropy for ternary PSK, ideal detectors; $B$ is the complete retained output of each receiver. (a)~Coherent-certified model: the four receiver families (solid, $d$ as in Table~\ref{tab:inner-outer-gap}) and the joint-output Hybrid-1, for which $B=(K,B_{\rm click})$ ($d=9$: beacon region $\times$ click bins $\{0\},\{1\},\{\ge2\}$; $g=4$, midpoint thresholds; open markers). The joint hybrid overtakes homodyne below $\mu\approx0.02$. (b)~Homodyne under the three source models: coherent Gram (certified), magnitude Gram (independent benchmark), and equal-overlap floor $\delta=(1-2\mu)_+$, which certifies nothing for $\mu\ge\tfrac12$.}
\label{fig:hmin}
\end{figure}

\begin{figure}[!h]
\centering
\includegraphics[width=\linewidth]{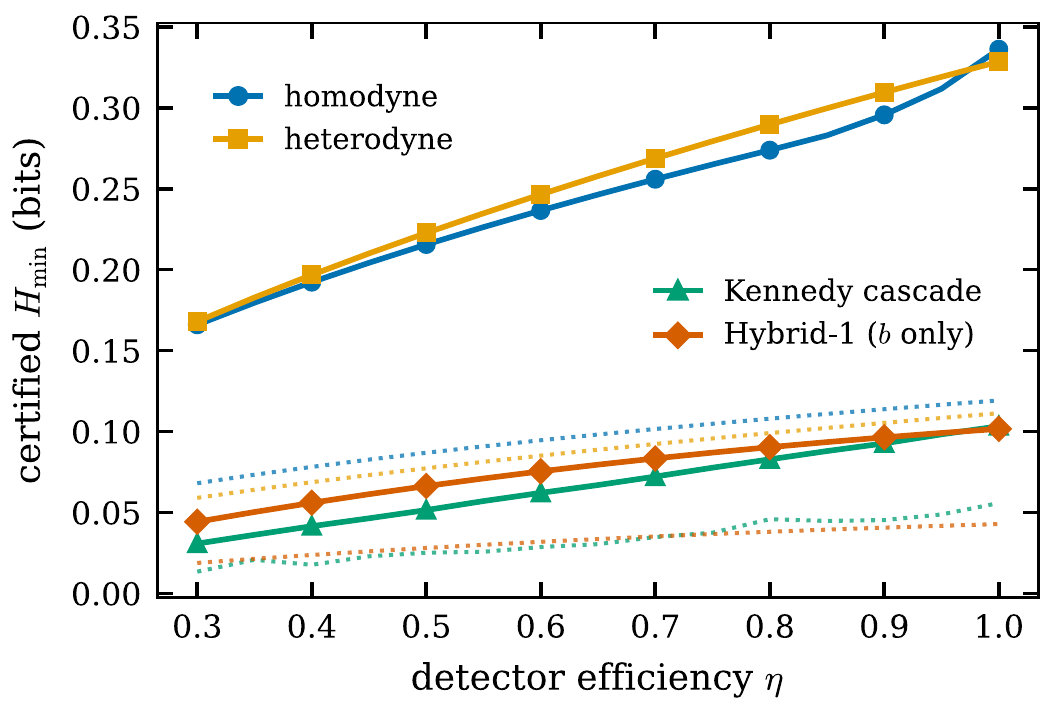}
\caption{Certified $H_{\min}$ against the effective detector efficiency $\eta$ at $\mu=0.2$, ternary PSK. Solid curves use the coherent-certified model and dotted curves the energy-floor benchmark. In the numerical model, $\eta$ is represented by the common pure-loss transformation $\alpha\rightarrow\sqrt{\eta}\alpha$ at the receiver input, while the internal receiver parameters remain fixed.}
\label{fig:eta}
\end{figure}

Certificate quality is reported rather than assumed. Figure~\ref{fig:certificates} shows the pointwise difference between the verified dual bound and the primal SDP value. Panel~(a) contains homodyne, heterodyne and coarse-grained Hybrid-1, whose differences all remain at least two orders of magnitude below the statistical width $\Delta_{b,x}$ and are therefore shown without this off-scale reference. Panel~(b) shows the Kennedy cascade separately because its nearly deterministic statistics produce larger and irregular primal--dual differences. At isolated points these differences exceed $\Delta_{b,x}$, so the reported entropy is obtained from the independently verified dual bound rather than the primal value. The joint-output Hybrid-1 is omitted from the figure but is checked through the same procedure, while Table~\ref{tab:certificates} gives the diagnostics at the reference point. The scans behind the floor model are shown in Fig.~\ref{fig:scan}.

\begin{figure}[!h]
\centering
\includegraphics[width=\linewidth]{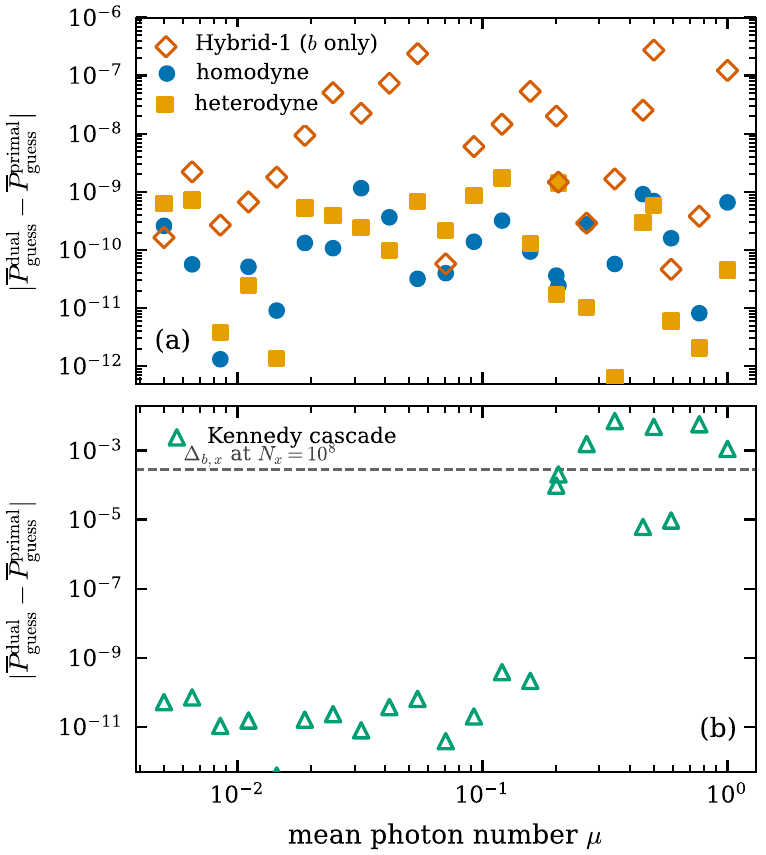}
\caption{Pointwise difference between the verified dual certificate and the primal SDP value in the coherent model. The markers are not connected because these are numerical termination diagnostics rather than a continuous physical quantity. Panel~(a) shows homodyne, heterodyne and coarse-grained Hybrid-1; all points lie at least two orders of magnitude below $\Delta_{b,x}$, which is off scale. Panel~(b) shows the Kennedy cascade together with the Hoeffding half-width at $N_x=10^{8}$ and $\varepsilon_{\rm stat}=10^{-6}$. At points above this line, the verified dual bound determines the reported entropy.}
\label{fig:certificates}
\end{figure}

\begin{figure}[!h]
\centering
\includegraphics[width=\linewidth]{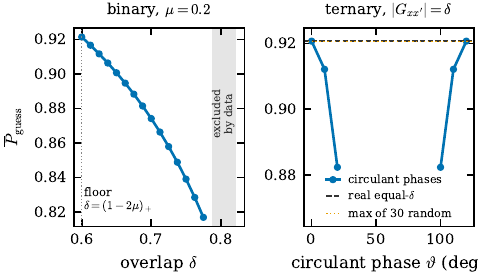}
\caption{Worst-case scans behind the energy-floor model at $\mu=0.2$ (homodyne data). (a)~Binary: $\overline P_{\rm guess}$ is non-increasing in the overlap over the feasible range, so the floor $\delta=(1-2\mu)_+$ is extremal; overlaps beyond the shaded edge are excluded by the data (no feasible realisation). (b)~Ternary at fixed modulus $|G_{xx'}|=\delta$: circulant phase families and $30$ random Hermitian phase patterns never exceed the real equal-overlap Gram.}
\label{fig:scan}
\end{figure}

\begin{table}[h]
\centering
\caption{Certification diagnostics at the reference point ($\mu=0.2$, ternary, coherent model): primal SDP value, a posteriori verified dual bound (Lemma~\ref{lem:dual}), their difference, worst eigenvalue residual of the verified certificate, and the effect of finite statistics ($N_x=10^8$, $\varepsilon_{\rm stat}=10^{-6}$). MOSEK 11.2.2, interior-point tolerances $10^{-10}$.}
\label{tab:certificates}
\begin{tabular}{lccccc}
\hline
Receiver & primal & dual cert. & gap & residual & $H_{\min}$ finite \\
\hline
Homodyne   & 0.79214 & 0.79214 & $4{\times}10^{-11}$ & $10^{-12}$ & 0.3346 \\
Heterodyne & 0.79637 & 0.79637 & $2{\times}10^{-11}$ & $10^{-13}$ & 0.3280 \\
Kennedy    & 0.93085 & 0.93095 & $1{\times}10^{-4}$ & $3{\times}10^{-6}$ & 0.0946 \\
Hybrid-1   & 0.93199 & 0.93199 & $2{\times}10^{-8}$ & $6{\times}10^{-10}$ & 0.1013 \\
\hline
\end{tabular}
\end{table}

\subsection{Receiver benchmarks}
For binary signalling the Kennedy receiver is the best practical decoder across the tested range: its mutual information exceeds two-bin homodyne by 8--16\% over $\mu\in[0.01,0.8]$ (Fig.~\ref{fig:kennedy-advantage}), with photon information efficiencies $\gamma=\lim_{\mu\to0}I/\mu$ of $\gamma_{\rm Ken}=2|t|^2=1.98$ against $\gamma_{\rm Hom}=4/(\pi\ln 2)\approx1.84$--the $8\%$ asymptotic edge growing at finite $\mu$ because of the extreme asymmetry of the Kennedy channel.

\begin{figure}[!h]
\centering
\includegraphics[width=\linewidth]{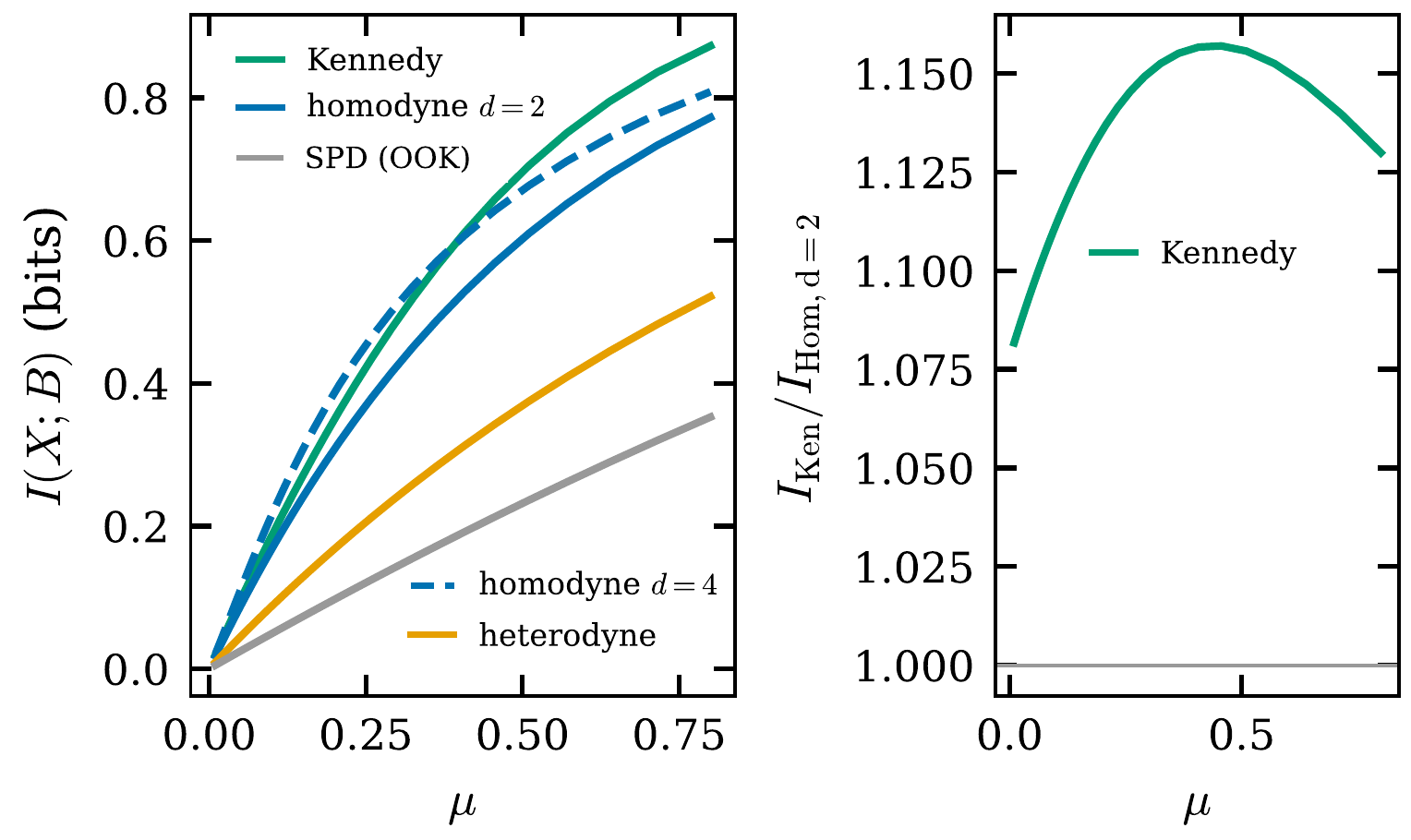}
\caption{Binary PSK. Left: $I(X;B)$ for Kennedy ($|t|^2=0.99$), homodyne with $d=2$ and $d=4$ bins, heterodyne, and SPD on--off keying (peak energy $\mu$). Right: Kennedy-to-homodyne ratio.}
\label{fig:kennedy-advantage}
\end{figure}

For ternary PSK the photon-efficiency ordering is $\gamma_{\rm Hyb}^{(k,b)}\approx1.10$ (joint output, $g=4$) $>\gamma_{\rm Hom}\approx0.97>\gamma_{\rm Het}\approx0.80>\gamma_{\rm Ken}\approx0.22$, all evaluated at $\mu=0.01$ with ideal detectors (Table~\ref{tab:gamma}, Fig.~\ref{fig:photon-efficiency}). The hybrid's edge over homodyne is $13\%$ at $\mu=0.01$ and grows to about $21\%$ as $\mu\to10^{-3}$; it needs the joint output--tracing out the beacon collapses the efficiency to $\gamma\approx0.56$--and it counts signal photons only, the displacement energy being auxiliary in the same sense as a local oscillator. The Kennedy cascade pays for its architectural degeneracy with $\gamma_{\rm Ken}\approx0.22$. For 5-PSK (Fig.~\ref{fig:5psk-hierarchy}), the joint hybrid has the largest photon efficiency at low energy, with $\gamma\approx1.12$ against $1.02$ for homodyne and $0.99$ for heterodyne at $\mu=0.01$; below $\mu\approx0.005$ the ordering of the two Gaussian receivers inverts, heterodyne approaching $0.99$ and homodyne $0.95$ as $\mu\to0$. Homodyne is the strongest non-hybrid receiver over the intermediate range until heterodyne overtakes at $\mu\approx0.73$, while the Kennedy cascade remains weaker with $\gamma\approx0.32$.

\begin{table}[h]
\centering
\caption{Photon information efficiency $\gamma_{\rm rec}=I(X;B)/\mu$ at $\mu=0.01$, ideal detectors, ternary PSK, for the stated output alphabets. The joint-output hybrid uses midpoint beacon thresholds and $g=4$.}
\label{tab:gamma}
\begin{tabular}{lc}
\hline
Receiver & $\gamma_{\rm rec}$ (bits/photon) \\
\hline
Hybrid-1, joint $(k,b)$, $g{=}4$ & 1.10 \\
Homodyne ($P$-quad, $d{=}3$)       & 0.97 \\
Heterodyne ($d{=}3$)               & 0.80 \\
Hybrid-1, $b$ only ($d{=}5$)       & 0.56 \\
Kennedy cascade ($d{=}3$, $|r|^2{=}0.1$) & 0.22 \\
\hline
\end{tabular}
\end{table}

\begin{figure}[!h]
\centering
\includegraphics[width=\linewidth]{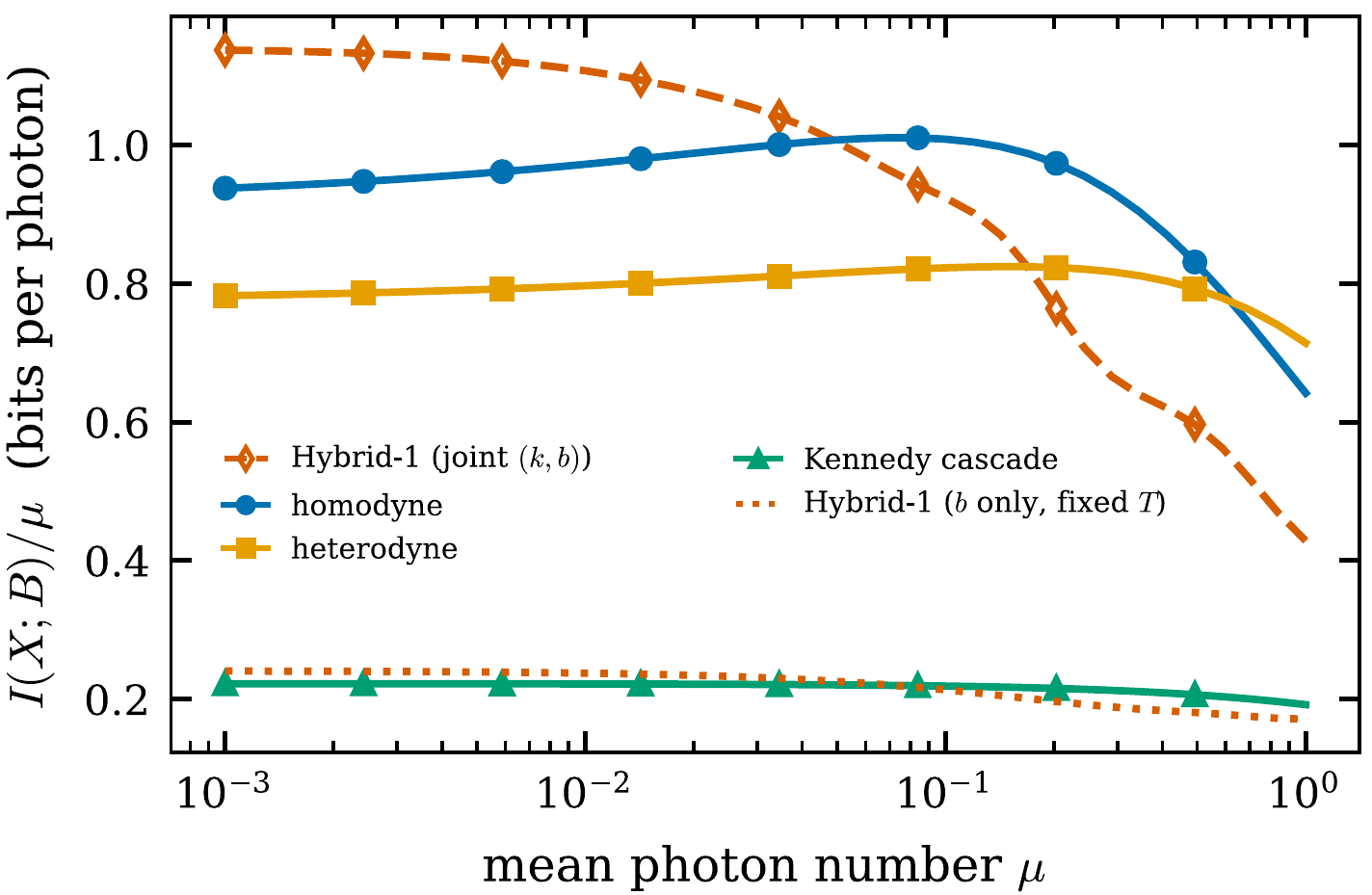}
\caption{Photon information efficiency $I(X;B)/\mu$, ternary PSK, ideal detectors. The joint-output hybrid ($g=4$, adaptive thresholds) is the only receiver above homodyne at low energy; the same receiver with the beacon label traced out (dotted) is not competitive, which quantifies the price of discarding the beacon record.}
\label{fig:photon-efficiency}
\end{figure}

\begin{figure}[!h]
\centering
\includegraphics[width=\linewidth]{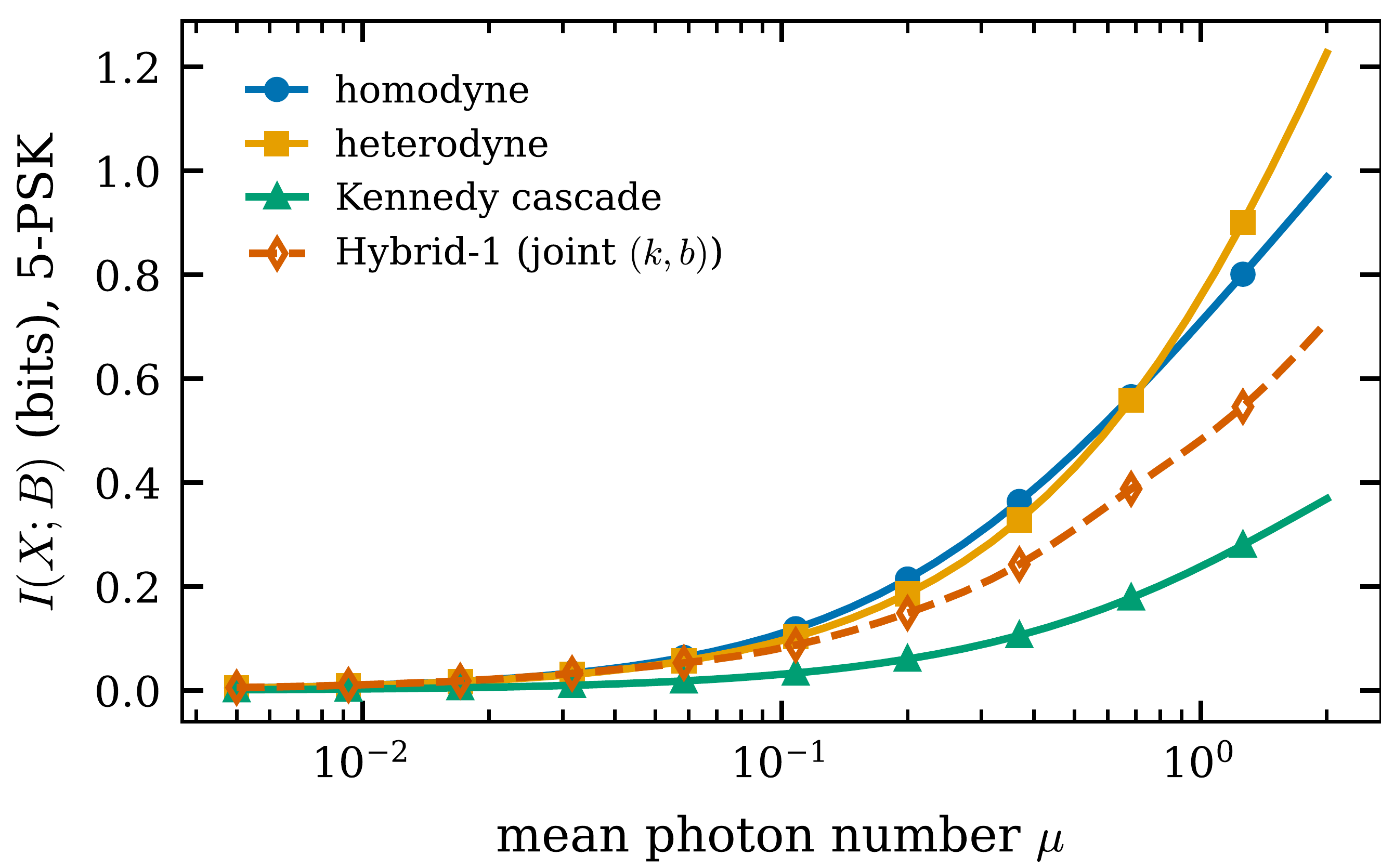}
\caption{$I(X;B)$ for 5-PSK. The joint-output hybrid has the largest slope in the extreme low-energy regime, after which homodyne with $d=5$ midpoint bins becomes dominant. Heterodyne with $d=5$ angular sectors overtakes homodyne at $\mu\approx0.73$, while the Kennedy cascade suffers the same one-null-per-stage limitation as in the ternary case.}
\label{fig:5psk-hierarchy}
\end{figure}

\section{Receiver-level implications}\label{sec:applications}
The object every protocol below inherits is the measured channel $P(b|x)$ under the stated fixed-Gram source model, so the receiver hierarchy can be used as a receiver-level benchmark; where a statement is phrased as key rate, covert throughput or forgery probability, it is a receiver benchmark induced by $P(b|x)$ under this paper's adversarial model, not a composable end-to-end proof, except where a full protocol analysis exists. The carrier is the mutual information
\begin{equation}
I(X;B)=H(B)-H(B|X),
\qquad p(b)=\tfrac1n\textstyle\sum_x P(b|x),
\label{eq:mutual-info}
\end{equation}
and, in the energy-starved limit, the single receiver constant
\begin{equation}
I(X;B)\;\approx\;\gamma_{\rm rec}\,\mu,\qquad \mu\to 0,
\label{eq:low-mu-scaling}
\end{equation}
already tabulated in Table~\ref{tab:gamma}. The quantity $\gamma_{\rm rec}$ therefore provides a common low-energy receiver benchmark for communication-oriented applications, although it does not determine the certified randomness rate, which is obtained independently from the guessing-probability SDP.

\subsection{Discrete-modulated CV-QKD}\label{sec:app-qkd}
In DM-CV-QKD with reverse reconciliation~\cite{Leverrier2009,Ghorai2019,Lin2019,Denys2021}, the asymptotic rate takes the protocol-dependent form
$K\ge\beta I(X;B)-\chi(B;E)$. We do not derive an end-to-end secret-key rate here; instead, $I(X;B)$ and $H_{\min}(B|X,\Lambda)$ are used only as receiver-level benchmarks under the classical-$\Lambda$ model. An adversary holding quantum correlations with Bob's mode lies outside the present certification. At $\mu_A=0.5$ and $0.2$~dB/km, the joint-output hybrid overtakes homodyne at $L\approx63$~km and is 17\% ahead at $100$~km (Fig.~\ref{fig:qkd-distance}), which is the low-$\mu$ ordering of Table~\ref{tab:gamma} expressed in distance. For binary links the Kennedy receiver is preferable at every tested $\mu$ when only two outcomes are available. Proposition~\ref{prop:phase-blind} carries over unchanged: a Bob station relying on photon counting alone yields input-independent statistics for fixed-modulus phase encoding and certifies zero entropy in this model, ruling out photon-counting-only receiver simplifications.

\begin{figure}[!h]
\centering
\includegraphics[width=\linewidth]{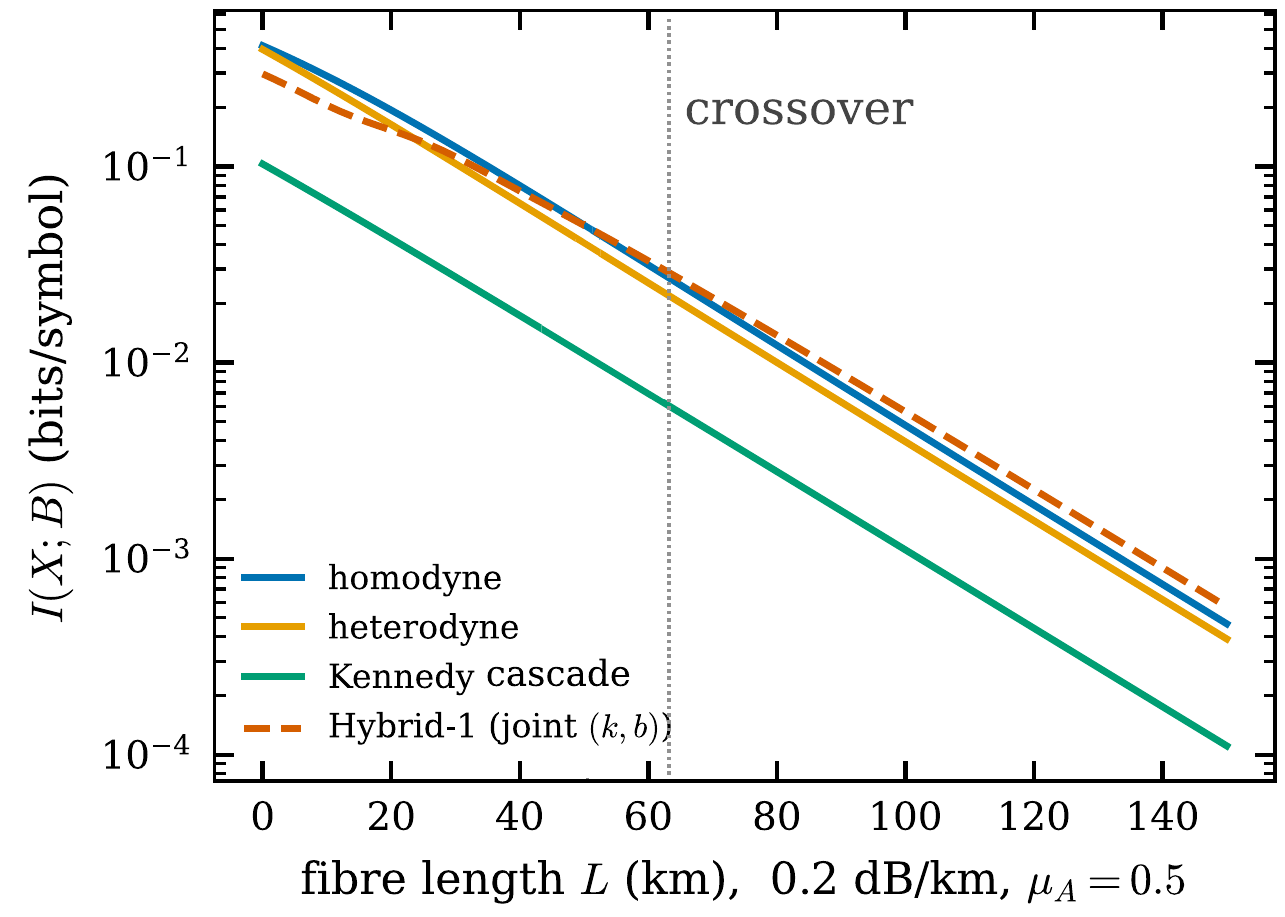}
\caption{$I(X;B)$ against fibre length for ternary PSK, $\mu_A=0.5$, $0.2$~dB/km. Homodyne dominates at short distance; the joint-output hybrid overtakes at $L\approx63$~km.}
\label{fig:qkd-distance}
\end{figure}

\subsection{Quantum reading and covert communication}
In quantum reading of a phase-encoded memory~\cite{Pirandola2011,ReadingCapacity2011} the per-probe information $I_{\rm read}(\mu)=I(X;B)$ benchmarks the probe budget: by Fano's inequality, any strategy identifying one of $n$ cells with error at most $\varepsilon$ from $N$ probe uses needs $N\,I_{\rm read}\ge(1-\varepsilon)\log_2 n-h_2(\varepsilon)$, so larger $I_{\rm read}$ is necessary for cheaper reading, while the achievable scaling in $\log(1/\varepsilon)$ is governed by the Chernoff exponent of $P(b|x)$ rather than by mutual information alone. Under a probe-energy cap the receiver ordering above applies to this benchmark verbatim, photon counting reading nothing from fixed-modulus phase cells. Covert communication over a bosonic channel obeys the square-root law $\mu\le c/\sqrt N$~\cite{Bash2015,Bullock2019}, so the total covert bits are $B\approx\gamma_{\rm rec}\,c\sqrt N$ by \eqref{eq:low-mu-scaling}: the receiver constant multiplies the covert throughput directly. At the operating point $\mu=0.01$ the joint hybrid carries 13\% more covert bits per $\sqrt N$ uses than homodyne for ternary signalling, and the Kennedy receiver the corresponding binary margin.
\begin{figure}[!h]
\centering
\includegraphics[width=\linewidth]{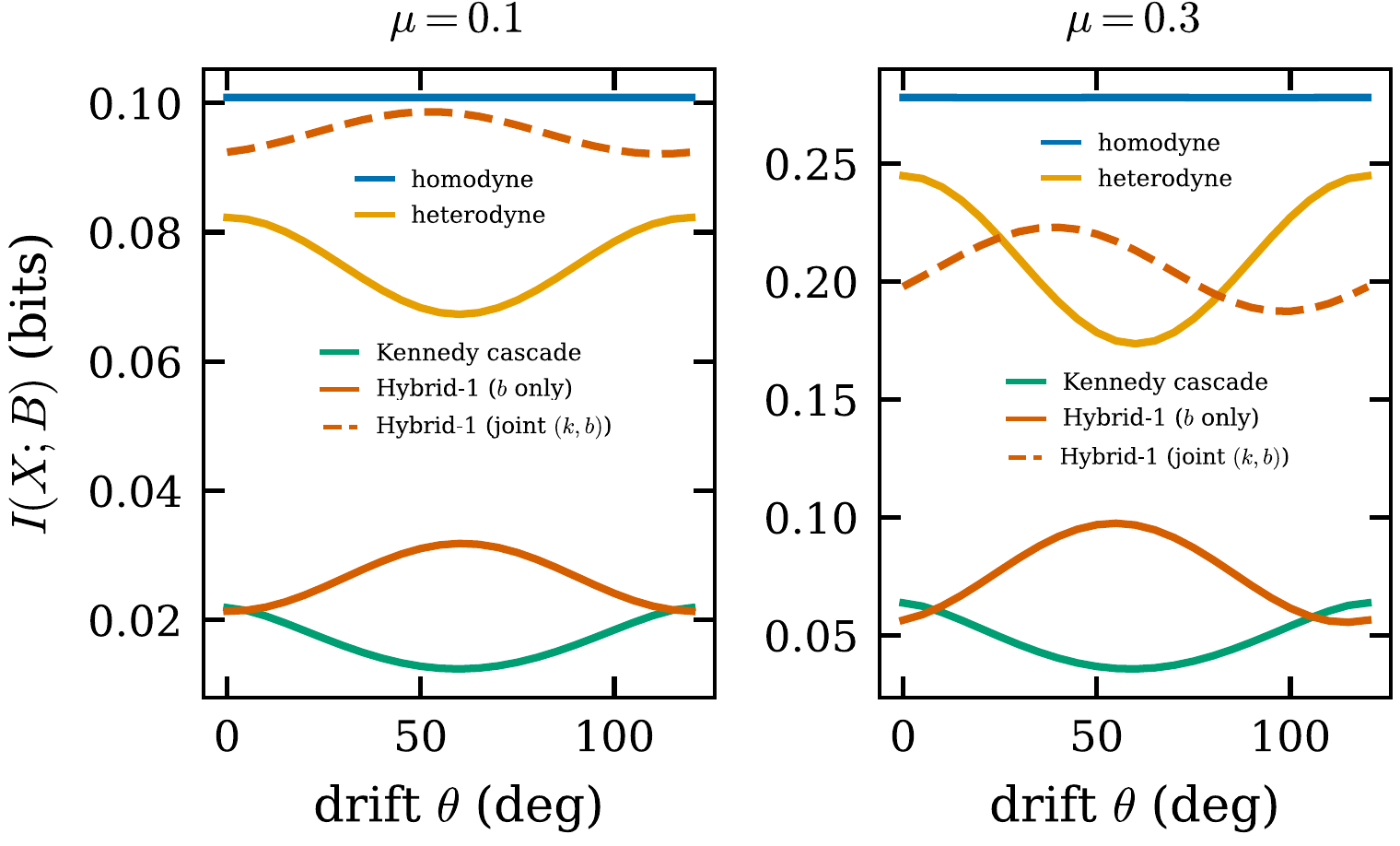}
\caption{$I(X;B)$ against LO-phase drift $\theta$ for ternary PSK at $\mu=0.1$ and $0.3$, all reference geometries frozen at their nominal ($\theta=0$) positions.}
\label{fig:phase-drift}
\end{figure}
\begin{figure*}[!t]
\centering
\includegraphics[width=\linewidth]{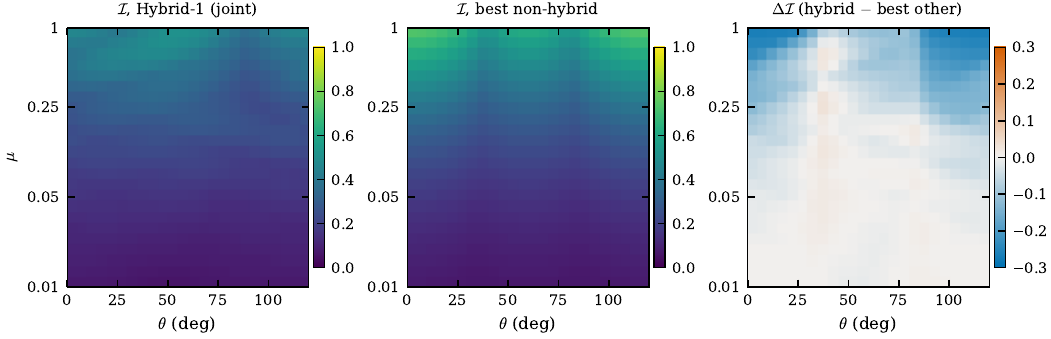}
\caption{Distinguishability index $\mathcal I$ of Eq.~\eqref{eq:distinguish-index} over $(\mu,\theta)$ for ternary PSK: joint-output Hybrid-1, best non-hybrid receiver, and their difference. The hybrid advantage is confined to low $\mu$ and survives constellation rotation.}
\label{fig:advantage_landscape}
\end{figure*}
\subsection{Quantum digital signatures}\label{sec:app-qds}
In coherent-state signature schemes~\cite{Croal2016,Thornton2019,Yin2023} the verification station measures phase-encoded coherent states. In the present model, the certified quantity $P_{\rm guess}(B|X,\Lambda)=2^{-H_{\min}(B|X,\Lambda)}$ characterises only the receiver-side predictability of the measurement outcome when the true input $X$ is supplied. We stress that this is a receiver-level predictability benchmark, not a forgery bound: converting it into a forgery probability requires the protocol-specific analysis--declaration format, mismatch thresholds, verification rule, the forger's actual information--of the scheme at hand, which we do not perform, and a forger entangled with the verification mode would in any case need $H_{\min}(B|X,E)$, which is not certified here. The calculated entropy should therefore be read only as a receiver-side randomness benchmark and does not imply an ordering of forgery difficulty across receivers. Within the present model, photon counting alone on phase-encoded states certifies no receiver-side unpredictability.

Beyond these, the certification primitive itself has uses we only note. Because the certified $H_{\min}$ and the inner--outer signature of Table~\ref{tab:inner-outer-gap} are computed from $P(b|x)$ alone, a client can benchmark whether a deployed receiver behaves consistently with a claimed channel model--the statistics cannot single out the hardware architecture, since by construction many realisations reproduce them, but an inconsistent claim is detectable; and a node whose source is certified at fixed Gram--binary energy-floor, or verified coherent PSK--plus any phase-sensitive receiver can run the scheme as a local certified-randomness service, Proposition~\ref{prop:phase-blind} acting as a hardware provisioning rule--photon-counting-only nodes cannot.

\subsection{Sensitivity to phase-reference drift}\label{sec:app-phase-drift}
A drift $\theta$ of the local-oscillator phase rotates the constellation against frozen bins, sectors, nulling phases and thresholds. For the Kennedy cascade with frozen nulling phases, the stage intensity becomes
\begin{equation}
|\gamma_{kx}(\theta)|^2
=4|r|^2\mu_k
\sin^2\!\left[\frac{2\pi(x-k)/n+\theta}{2}\right],
\end{equation}
so its response also changes with $\theta$. In mutual information the computed picture (Fig.~\ref{fig:phase-drift}) confirms the expected ordering with a twist, i.e., frozen-bin homodyne remains very close to its nominal value throughout the ternary period; at $\mu=0.3$ and $\theta=\pi/6$, its mutual information changes from $0.278021$ to $0.277896$~bits, whereas frozen-sector heterodyne loses up to $29\%$ at $\theta=60^\circ$, $\mu=0.3$ ($18\%$ at $\mu=0.1$), so at the level of $I(X;B)$ it is heterodyne that needs sector tracking. Mutual information is, however, the wrong drift metric for a randomness generator. The two metrics decouple at the merge angle: at $\mu=0.3$ and $\theta=\pi/6$, the homodyne columns for $x=0$ and $x=1$ become identical, $I(X;B)$ changes only slightly from $0.278021$ to $0.277896$~bits, while the certified $H_{\min}$ drops from $0.295$ to $0.213$~bits, because partial input-independence hands the adversary a deterministic branch on the merged pair. Drift budgets for certification must therefore be set on $H_{\min}$, where homodyne's apparent immunity does not survive. The hybrid sits between the two receivers in both metrics, with its beacon providing the phase-sensitive record needed for adaptive tracking. The Kennedy cascade also degrades under drift when the nulling phases are fixed. Figure~\ref{fig:advantage_landscape} maps the full $(\mu,\theta)$ landscape of the calibrated index for the joint hybrid against the best non-hybrid receiver.

At experimentally realistic parameters ($\eta=0.85$, $p_{\rm dc}=10^{-5}$, $N_x=10^8$, $\mu=0.2$, ternary), homodyne certifies $0.283$~bits against $0.336$~bits in the ideal case, with a further finite-statistics reduction of $1.6\times10^{-3}$~bits. Figure~\ref{fig:eta} shows that heterodyne is slightly higher at this transmissivity. The effect of energy calibration on the certified entropy is evaluated separately at fixed observed data, as discussed above.

\section{Conclusion}\label{sec:conclusion}
We have benchmarked five optical receiver families through the conditional laws $P(b|x)$ they generate, using semi-device-independent randomness certification as the working demonstration. For fixed-Gram source models, the adversarial guessing problem is an exact SDP, including the complex coherent Gram for $n\ge3$ through the block-real embedding, and every reported value is supported by an independently checked dual certificate. The coherent-certified model provides the main results, while the magnitude Gram is treated as an independent benchmark and the energy-derived overlap floor gives a certified worst case for binary alphabets. The model-free energy-only treatment remains valid but becomes too weak at practical Fock truncations. Since the fixed-Gram program is exact, the difference between honest and adversarial guessing measures the value of classical side information rather than numerical relaxation slack, reaching about half a bit for the phase-sensitive ternary receivers at $\mu=0.2$.

The receiver ordering depends strongly on both the signal energy and the retained output alphabet. Photon counting is phase blind for fixed-modulus PSK and therefore certifies no randomness without an additional phase-sensitive element. Homodyne with MAP binning gives the largest certified entropy at moderate energy, reaching $0.336$~bits per round at $\mu=0.2$ in the coherent model and remaining comparatively robust to loss. At low energy, the joint-output hybrid overtakes homodyne under the nominal source calibration when the beacon-region label is retained. The advantage survives finite statistics at $N_x=10^8$, although the remaining margin is comparable to the change caused by a $1\%$ energy-calibration uncertainty. Removing this label eliminates both the certified-entropy and photon-efficiency advantages. The same conditional laws provide receiver-level comparisons for discrete-modulated CV-QKD, quantum reading, covert communication and quantum-signature verification, although they do not replace the protocol-specific security analyses required in those settings. The phase-drift results further show that mutual information and certified entropy need not respond in the same way, so stability requirements for randomness certification must be set directly on $H_{\min}$.
The security analysis is restricted to classical side information, and the multi-round extraction statement assumes collective i.i.d.\ attacks. Extending the certification to quantum side information or devices with memory would require an instrument-based hierarchy, entropy accumulation or a related argument. The $n\ge3$ energy-floor values remain worst-case benchmarks because the phase scans support, but do not prove, the extremality of the real equal-overlap Gram. The results for applications other than randomness should likewise be read as receiver benchmarks derived from $P(b|x)$ rather than composable protocol proofs. Within these assumptions, the same optimisation and certificate procedure can be applied directly to other optical receivers.
\newpage
\bibliography{ref_updated}

\end{document}